\documentclass[envcountsame]{llncs}
\pdfoutput=1 

\usepackage{amssymb}
\usepackage{amsmath}
\usepackage{url}
\usepackage{wrapfig}
\usepackage{graphics}
\usepackage{graphicx}

\newcommand{\bfd}{\mathbf{d}}

\newcommand{\oldcomment}[1]{}

\newcommand\tuple[1]{\langle #1 \rangle}
\newcommand\etal{{\it et al{.}}}
\newcommand\minisatp{\textsc{MiniSat$^+$}}
\newcommand{\Base}{\mathit{Base}}
\newcommand{\intMultiSet}{\mathit{ms(\mathbb{N})}}
\newcommand{\sumCarry}{\mathit{sum\_carry}}
\newcommand{\sumDigits}{\mathit{sum\_digits}}
\newcommand{\numComparators}{\mathit{num\_comp}}
\newcommand{\cost}{\mathit{cost}}

\newenvironment{SProg}
     {\begin{small}\begin{tt}\begin{tabular}[c]{l}}%
     {\end{tabular}\end{tt}\end{small}}

\newenvironment{SProg2}
     {\begin{small}\begin{tt}\begin{tabular}[c]{ll}}%
     {\end{tabular}\end{tt}\end{small}}
\newcommand{\qin}{\hspace*{0.15in}}

\newcommand{\set}[1]{\left\{
      \begin{array}{l}#1\end{array}
        \right\}}
\newcommand{\sset}[2]{\left\{~#1  \left|
      \begin{array}{l}#2\end{array}
    \right.     \right\}}
\newcommand{\card}[1]{ 
    \left|
      \begin{array}{l}#1\end{array}
    \right|}
\def\anno#1{{\ooalign{\hfil\raise.07ex\hbox{\small{\rm #1}}\hfil%
        \crcr\mathhexbox20D}}}
\newcommand{\ceil}[1]{\lceil{#1}\rceil}

\title{Optimal Base Encodings for\\
       Pseudo-Boolean Constraints\thanks{Supported by GIF\ grant
       966-116.6 and the Danish Natural Science Research Council.}
       }

\author{    Michael Codish\inst{1}
            \and
            Yoav Fekete\inst{1}
            \and
            Carsten Fuhs\inst{2}
            \and
            Peter Schneider-Kamp\inst{3}
}%

\institute{
	Department of Computer Science, 
        Ben Gurion University of the Negev, Israel
	\and 	  
	LuFG Informatik 2, RWTH Aachen University, Germany 
	\and 	  
        IMADA, University of Southern Denmark, Denmark
}%

\begin{document}
\maketitle

\begin{abstract}
  This paper formalizes the \emph{optimal base problem}, presents an
  algorithm to solve it, and describes its application to the encoding
  of Pseudo-Boolean constraints to SAT.  We demonstrate the impact of
  integrating our algorithm within the Pseudo-Boolean constraint
  solver \minisatp.  Experimentation indicates that our algorithm
  scales to bases involving numbers up to 1,000,000, improving on the
  restriction in \minisatp\ to prime numbers up to 17.  We show that,
  while for many examples primes up to 17 do suffice, encoding with
  respect to optimal %
  bases 
  reduces the CNF sizes and improves the subsequent SAT solving
  time for many examples.
\end{abstract}

\section{Introduction}

The optimal base problem is all about finding an efficient
representation for a given collection of positive integers. One
measure for the efficiency of such a representation is the sum of the
digits of the numbers.  Consider for example the decimal numbers
$S=\{16,30,54,60\}$. The sum of their digits is 25.  Taking binary
representation we have $S_{(2)}=\{10000,11110,110110,$ $111100\}$ and
the sum of digits is 13, which is smaller. Taking ternary
representation gives $S_{(3)}=\{121,1010,2000,2020\}$ with an even
smaller sum of digits, 12. Considering the \emph{mixed radix} base
$B=\tuple{3,5,2,2}$, the numbers are represented as $S_{(B)}=\{101,$
$1000,1130,10000\}$ and the sum of the digits is 9.  The optimal base
problem is to find a (possibly mixed radix) base for a given sequence
of numbers to minimize the size of the representation of the
numbers. When measuring size as ``sum of digits'', the base $B$ is
indeed optimal for the numbers of $S$. In this paper we present the
optimal base problem and illustrate why it is relevant to the encoding
of Pseudo-Boolean constraints to SAT.
We also present an algorithm and show that our implementation is
superior to current implementations.

Pseudo-Boolean constraints take the form $a_1x_1 + a_2x_2 + \cdots +
a_nx_n \geq k$, where $a_1,\ldots, a_n$ are integer coefficients,
$x_1,\ldots,x_n$ are Boolean literals (i.e., Boolean variables or
their negation), and $k$ is an integer.  We assume that constraints
are in Pseudo-Boolean normal form~\cite{Barth95}, that is, the
coefficients $a_i$ and $k$ are always positive and Boolean
variables occur at most once in $a_1x_1 + a_2x_2 + \cdots + a_nx_n$.
Pseudo-Boolean constraints are well studied and arise in many
different contexts, for example in verification \cite{Bryant02} and in
operations research \cite{Bixby92}.  Typically we are interested in
the satisfiability of a conjunction of Pseudo-Boolean constraints.
Since 2005 there is a series of Pseudo-Boolean Evaluations
\cite{Manquinho06} 
which aim to assess the state of the art in the field of
Pseudo-Boolean solvers. We adopt these competition problems as a
benchmark for the techniques proposed in this paper.

Pseudo-Boolean constraint satisfaction problems are often reduced to
SAT. 
Many works describe techniques to encode these constraints to
propositional formulas \cite{BailleuxBR06,BailleuxBR09,EenS06}.
The Pseudo-Boolean solver \minisatp\ (\cite{EenS06}, cf.\
\url{http://minisat.se}) chooses between three techniques to
generate SAT encodings for Pseudo-Boolean constraints. These convert
the constraint to: (a) a BDD structure, (b) a network of sorters, and
(c) a network of (binary) adders. The network of adders is the most
concise encoding, but it has the weakest propagation properties and
often leads to higher SAT solving times
than the BDD based encoding, which, on the other hand,
generates the largest encoding. The encoding based on sorting networks
is often the one applied and is the one we consider in this paper.

\begin{wrapfigure}{r}{42mm}\vspace{-10mm}
  \begin{center}\scriptsize  
     $\begin{array}{ll|l|ll}
      \cline{3-3}
      x_5    & ~ &\hspace{8mm} & ~& y_8\\
        \cline{2-2} \cline{4-4}
      x_5 & ~ & & & y_7\\
        \cline{2-2} \cline{4-4}
      x_5 & ~ & & & y_6\\
        \cline{2-2} \cline{4-4}
      x_4 & ~ & & & y_5\\
        \cline{2-2} \cline{4-4}
      x_4 & ~ & & & y_4=1\\
        \cline{2-2} \cline{4-4}
      x_3 & ~ & & & y_3=1\\
        \cline{2-2} \cline{4-4}
      x_2 & ~ & & & y_2=1\\
        \cline{2-2} \cline{4-4}
      x_1 & ~ & & & y_1=1\\
        \cline{2-2} \cline{4-4}
          & & &\\
      \cline{3-3}
    \end{array}$
  \end{center}
  \vspace{-11mm}
\end{wrapfigure}
To demonstrate how sorters can be used to translate Pseudo-Boolean
constraints, consider the constraint $\psi=x_1 + x_2 + x_3 + 2x_4 +
3x_5 \geq 4$ where the sum of the coefficients is 8. On the right, we
illustrate an $8\times 8$ sorter where $x_1,x_2,x_3$ are each fed into
a single input, $x_4$ into two of the inputs, and $x_5$ into three of
the inputs. The outputs are $y_1,\ldots,y_8$.  First, we represent the
sorting network as a Boolean formula, $\varphi$, which in general, for
$n$ inputs, will be of size $O(n\log^2 n)$ \cite{Batcher68}. Then, to
assert $\psi$ we take the conjunction of $\varphi$ with the formula
$y_1\land y_2\land y_3\land y_4$.

But what happens if the coefficients in a constraint are larger than
in this example?  How should we encode 
$16 x_1 + 30 x_2 + 54 x_3 + 30 x_4 + 60 x_5 \geq 87$?  How should we
handle very large coefficients (larger than 1,000,000)? To this end,
the authors in~\cite{EenS06} generalize the above idea and propose 
to decompose the constraint into a number of interconnected
sorting networks. Each sorter represents a digit in a mixed radix
base.
This construction is governed by the choice of a suitable mixed radix
base and the objective is to find a base which minimizes the size of
the sorting networks. 
Here the optimal base problem comes in, as the size of the
networks is directly related to the size of the representation of the
coefficients. We consider the sum of the digits (of coefficients)
and other measures for the size of the representations
and study their influence on the quality of the encoding.

In \minisatp\ the search for an optimal base is performed using a
brute force algorithm and the resulting base is constructed from prime
numbers up to 17. The starting point for this paper is the following
remark	 from \cite{EenS06} (Footnote 8):
\vspace{-1mm}
\begin{quotation}\noindent
  \emph{This is an ad-hoc solution that should be improved in the
  future. Finding the optimal base is a challenging optimization
  problem in its own right.}
\end{quotation}
In this paper we take the challenge and present an algorithm which
scales to find an optimal base consisting of elements 
with values up to 1,000,000. We illustrate that in many
cases finding a better base leads also to better SAT solving time.

\pagebreak

Section~\ref{section:obp} provides preliminary definitions and
formalizes the optimal base problem.
Section~\ref{sec:encoding} describes how \minisatp decomposes a
Pseudo-Boolean constraint with respect to a given mixed radix base to
generate a corresponding propositional encoding, so that the constraint has
a solution precisely when the encoding has a model.
Section~\ref{section:4} is about (three) alternative measures with
respect to which an optimal base can be found.
Sections~\ref{sec:ob1}--\ref{sec:ob3} introduce our algorithm
based on classic AI search methods (such as cost underapproximation)
in
three steps: Heuristic pruning, best-first branch and bound, and base
abstraction.  
Sections~\ref{sec:exp} and \ref{relwork} present an experimental
evaluation and some related work.  
Section~\ref{sec:conc} concludes. 
Proofs are given in the appendix.

\section{Optimal Base Problems}
\label{section:obp}

In the classic base $r$ radix system, positive integers are
represented as finite sequences of digits
$\bfd=\tuple{d_0,\ldots,d_k}$ where for each digit $0\leq d_i<r$, and
for the most significant digit, $d_k>0$.  The integer value associated
with $\bfd$ is $v=d_0 + d_1r + d_2r^2+\cdots+d_kr^k$.
A mixed radix system is a generalization where a base is an infinite
radix sequence $B=\tuple{r_0,r_1,r_2,\ldots}$ of integers where for
each radix, $r_i>1$ and for each digit, %
$0\leq d_i < r_i$.
The integer value associated with $\bfd$ is 
$v = d_0w_0+ d_1w_1 + d_2w_2+\cdots+d_kw_k$
where $w_0 = 1$ and for $i\ge 0$, $w_{i+1} = w_ir_i$.
The sequence $weights(B)=\tuple{w_0,w_1,w_2,\ldots}$ specifies the
weighted contribution of each  digit position and is called the
\emph{weight sequence of $B$}.  
A finite mixed radix base is a finite sequence
$B=\tuple{r_0,r_1,\ldots,r_{k-1}}$ with the same restrictions as for
the infinite case except that numbers always have $k+1$ digits
(possibly padded with zeroes) and there is no bound on the value of
the most significant digit, $d_k$.

In this paper we focus on the representation of finite multisets of
natural numbers in finite mixed radix bases.  Let $\Base$ denote the
set of finite mixed radix bases and $\intMultiSet$ the set of finite
non-empty multisets of natural numbers. We often view multisets as
ordered (and hence refer to their first element, second element, etc.).
For a finite sequence or multiset $S$ of natural numbers, we denote
its length by $|S|$, its maximal element by $max(S)$, its $i^{th}$
element by $S(i)$, and the multiplication of its elements by $\prod
S$ (if $S$ is the empty sequence then $\prod S=1$).
If a base consists of prime numbers only, then we say that it is a
prime base. The set of prime bases is denoted $\Base_p$.

Let $B\in\Base$ with $|B|=k$. We denote by $v_{(B)}=\tuple{d_0, d_1,
  \ldots, d_k}$ the representation of a natural number $v$ in base
$B$. The most significant digit
 of $v_{(B)}$, denoted $msd(v_{(B)})$, is
$d_k$.  If $msd(v_{(B)})=0$ then we say that $B$ is redundant for $v$.
Let $S\in\intMultiSet$ with $|S|=n$. We denote the $n\times (k+1)$
matrix of digits of elements from $S$ in base $B$ as
$S_{(B)}$. Namely, the $i^{th}$ row in $S_{(B)}$ is the vector
$S(i)_{(B)}$. The most significant digit column of $S_{(B)}$ is the
$k+1$ column of the matrix and denoted $msd(S_{(B)})$.  If
$msd(S_{(B)})=\tuple{0,\ldots,0}^T$, then we say that $B$ is redundant
for $S$.
This is equivalently characterized by $\prod B > max(S)$.

\begin{definition}[non-redundant bases]
\label{def:nrb}
  Let $S\in\intMultiSet$. We denote the set of non-redundant bases for
  $S$, $\Base(S) = \sset{B\in\Base}{\prod B \leq max(S)}$. The set of  
  non-redundant prime bases for $S$ is denoted $\Base_p(S)$. 
  \pagebreak
  The set of non-redundant (prime) bases for $S$, containing elements
  no larger than $\ell$, is denoted $\Base^\ell(S)$
  ($\Base_p^\ell(S)$). 
  The set of bases in $\Base(S)$/$\Base^\ell(S)$/$\Base_p^\ell(S)$, is
  often viewed as a tree with root $\tuple{~}$ (the empty base) and an
  edge from $B$ to $B'$ if and only if $B'$ is obtained from $B$ by
  extending it with a single integer value.
\end{definition}

\begin{definition}[$\sumDigits$]\label{def:sumDigits}
  Let $S\in\intMultiSet$ and $B\in\Base$.  The sum of the digits of
  the numbers from $S$ in base $B$ is denoted $\sumDigits(S_{(B)})$.
\end{definition}

\begin{example}
  The usual binary ``base 2'' and ternary ``base 3'' are represented
  as the infinite sequences $B_1=\tuple{2,2,2,\ldots}$ and
  $B_2=\tuple{3,3,3,\ldots}$. The finite sequence
  $B_3=\tuple{3,5,2,2}$ and the empty sequence $B_4=\tuple{~}$ are
  also bases. The empty base is often called the ``unary base''
  (every number in this base has a single digit).  
  Let $S = \{16,30,54,60\}$.  Then, 
  $\sumDigits(S_{(B_1)})=13$,~~
  $\sumDigits(S_{(B_2)})=12$, ~~
  $\sumDigits(S_{(B_3)})=9$, ~and
  $\sumDigits(S_{(B_4)})=160$.
\end{example}

Let $S\in\intMultiSet$. A cost function for $S$ is a function
$\cost_S:\Base\rightarrow\mathbb{R}$ which associates bases with real
numbers. An example is $\cost_S(B)=\sumDigits(S_{(B)})$.
In this paper we are concerned with the following \emph{optimal base
problem}. 

\begin{definition}[optimal base problem]
  Let $S \in \intMultiSet$ and $\cost_S$ a cost function.  We say that
  a base $B$ is an \emph{optimal base for $S$} with respect to
  $\cost_S$, if for all bases $B'$, $\cost_S(B) \leq \cost_S(B')$.
  The corresponding \emph{optimal base problem} is to find an optimal
  base $B$ for $S$.
\end{definition}

The following two lemmata confirm that for the $\sumDigits$ cost
function, we may restrict attention to non-redundant bases involving
prime numbers only.

\begin{lemma}
\label{l1}
   Let $S \in \intMultiSet$ and consider the $\sumDigits$ cost
   function. Then, $S$ has an optimal base in $\Base(S)$.
\end{lemma}

\begin{lemma}
\label{lem:primes}
   Let $S \in \intMultiSet$ and consider the $\sumDigits$ cost
   function. Then, $S$ has an optimal base in $\Base_p(S)$.
\end{lemma}

How hard is it to solve an instance of the optimal base problem
(namely, for $S\in\intMultiSet$)? The following lemma provides a
polynomial (in $max(S)$) upper bound on the size of the search
space. This in turn suggests a pseudo-polynomial time brute force
algorithm (to traverse the search space).

\begin{lemma}\label{zeta}
  Let $S\in\intMultiSet$  with $m=max(S)$. Then,
  $\card{\Base(S)} \leq m^{1+\rho}$ where
  $\rho=\zeta^{-1}(2)\approx 1.73$ and where $\zeta$ is the Riemann
  zeta function.
\end{lemma}

\begin{proof}
  Chor \etal\ prove in \cite{ChorLM00} that the number of ordered
  factorizations of a natural number $n$ is less than $n^{\rho}$. The
  number of bases for all of the numbers in $S$ is hence bounded by
  $\sum_{n\leq m}n^{\rho}$, which is bounded by $m^{1+\rho}$.
\end{proof}

\section{Encoding Pseudo-Boolean Constraints}
\label{sec:encoding}

This section presents the construction underlying the sorter based
encoding of Pseudo-Boolean constraints applied in \minisatp
\cite{EenS06}.  It is governed by the choice of a mixed radix base
$B$, the optimal selection of which is the topic of this paper. 
The construction sets up a series of sorting networks to encode the 
\pagebreak
digits, in base $B$, of the sum of the terms on the left side of a 
constraint $\psi = a_1x_1 + a_2x_2 + \cdots + a_nx_n \geq k$. The
encoding then compares these digits with those from $k_{(B)}$ from the
right side.
We present the construction, step by step, through an example where
$B=\tuple{2,3,3}$ and $\psi = 2x_1 + 2x_2 + 2x_3 + 2x_4 + 5x_5 + 18x_6
\geq 23$. 

\vspace{-2mm}
\paragraph{\textbf{Step one - representation in base:}}~\\
\vspace{-3mm}
\begin{wrapfigure}{r}{30mm}\vspace{-17mm}
\[\scriptsize
     S_{(B)}=\left(
       \begin{array}{lllll}
           0 & 1 & 0 & 0 \\
           0 & 1 & 0 & 0 \\
           0 & 1 & 0 & 0 \\
           0 & 1 & 0 & 0 \\
           1 & 2 & 0 & 0 \\
           0 & 0 & 0 & 1 \\
       \end{array}
     \right)
\]
\vspace{-13mm}
\end{wrapfigure}
\noindent
The coefficients of $\psi$ form a multiset $S= \{2, 2, 2, 2, 5, 18\}$
and their representation in base $B$, a $6\times 4$ matrix,
$S_{(B)}$, depicted on the right. The  rows of the matrix
correspond to  the representation of the coefficients
 in base $B$.

\vspace{-2mm}
\paragraph{\textbf{Step two - counting:}}

Representing the coefficients as four digit numbers in base
$B=\tuple{2,3,3}$ and considering the values
$weights(B)=\tuple{1,2,6,18}$ of the digit positions, we obtain a
decomposition for the left side of $\psi$:
\begin{eqnarray*}
 \lefteqn{2x_1 + 2x_2 + 2x_3 + 2x_4 + 5x_5 + 18x_6 =  }\\
   &&\mathbf{1}\cdot (x_5) + \mathbf{2}\cdot(x_1+x_2+x_3+x_4+2x_5) + 
 \mathbf{6}\cdot(0) + \mathbf{18}\cdot(x_6)
\end{eqnarray*}
To encode the sums at each digit position ($1,2,6,18$), we set up a
series of four  sorting networks as depicted below.
Given values for the variables, the sorted 
\begin{wrapfigure}{r}{80mm}\vspace{-10mm}
\[
\begin{array}{cccc}
\tiny    \begin{array}{ll|l|l}
      \cline{3-3}
          & ~ &\hspace{4mm} & ~\\
          & ~ &\hspace{4mm} & ~\\
          & ~ &\hspace{4mm} & ~\\
          & ~ & & \\
          & ~ & & \\
      x_5 & ~ & & \\
        \cline{2-2} \cline{4-4}
          & & &\\
      \cline{3-3}
    \end{array}
~~& %
\tiny    \begin{array}{ll|l|l}
        \cline{3-3}
     x_1     & ~ &\hspace{4mm} & ~\\
         \cline{2-2} \cline{4-4}
     x_2     & ~ &\hspace{4mm} & ~\\
         \cline{2-2} \cline{4-4}
     x_3     & ~ &\hspace{4mm} & ~\\
         \cline{2-2} \cline{4-4}
      x_4 & ~ & & \\
        \cline{2-2} \cline{4-4}
      x_5 & ~ & & \\
        \cline{2-2} \cline{4-4}
      x_5 & ~ & & \\
        \cline{2-2} \cline{4-4}
          & & &\\
      \cline{3-3}
    \end{array}
~~& %
\tiny    \begin{array}{ll|l|l}
      \cline{3-3}
          & ~ &\hspace{4mm} & ~\\
          & ~ &\hspace{4mm} & ~\\
          & ~ &\hspace{4mm} & ~\\
          & ~ & & \\
          & ~ & & \\
          & ~ & & \\
          & & &\\
      \cline{3-3}
    \end{array}
~~& %
\tiny    \begin{array}{ll|l|l}
      \cline{3-3}
          & ~ &\hspace{4mm} & ~\\
          & ~ &\hspace{4mm} & ~\\
          & ~ &\hspace{4mm} & ~\\
       & ~ & & \\
       & ~ & & \\
      x_6 & ~ & & \\
        \cline{2-2} \cline{4-4}
          & & &\\
      \cline{3-3}
    \end{array}
\\[8mm]
~~~~count~1's & ~~~~count~2's & ~~~~count~6's& ~~~~count~18's 
\end{array}
\]\vspace{-12mm}
\end{wrapfigure}
outputs from these
networks represented unary numbers $d_1,d_2,d_3,d_4$ such that the left
side of $\psi$ takes the value $ 1\cdot d_1 + 2\cdot d_2 + 6\cdot d_3
+ 18\cdot d_4$.

\paragraph{\textbf{Step three - converting to base:}}
For the  outputs $d_1,d_2,d_3,d_4$ to represent the
digits of a number in base $B=\tuple{2,3,3}$, we need to encode
also the ``carry'' operation from each digit position to the next.
The first 3 outputs must represent
valid digits for $B$, i.e., unary numbers less than $\tuple{2,3,3}$
respectively. In our example the single potential violation to this
restriction is $d_2$, which is represented in 6 bits.  To this end we
add two components to the encoding: (1) each third output of the
second network ($y_3$ and $y_6$ in the diagram) is fed into the third
network as an additional (carry) input; and (2) clauses are added to
encode that the output of the second network is to be considered
modulo 3. We call these additional clauses a \emph{normalizer}. 
The normalizer 
\begin{wrapfigure}{r}{83mm}\vspace{-10mm}
\[\begin{array}{cccc}
\tiny    \begin{array}{ll|l|l}
      \cline{3-3}
          & ~ &\hspace{4mm} & ~\\
          & ~ &\hspace{4mm} & ~\\
          & ~ &\hspace{4mm} & ~\\
          & ~ & & \\
          & ~ & & \\
      x_5 & ~ & & \\
        \cline{2-2} \cline{4-4}
          & & &\\
      \cline{3-3}
    \end{array}
~~& %
\tiny    \begin{array}{ll|l|l|l|ll}
      \cline{3-3} \cline{5-5}
      x_1    & ~ &\hspace{4mm} & ~y_6~  &\hspace{2mm}& ~ &y_6 \\
          \cline{2-2} \cline{4-6}   
      x_2    & ~ &\hspace{4mm} & ~y_5 &\hspace{2mm}& ~ & \\
          \cline{2-2} \cline{4-4} 
      x_3    &   &             & ~y_4 &         &   & \\
        \cline{2-2} \cline{4-4}  
      x_4 & ~ & &~y_3 & & & y_3\\
        \cline{2-2} \cline{4-6} 
      x_5 & ~ & &~y_2 & & &r_1\\
        \cline{2-2} \cline{4-4} \cline{6-6}
      x_5 & ~ & &~y_1 & & &r_2\\
        \cline{2-2} \cline{4-4} \cline{6-6}
          & & & & & & \\
      \cline{3-3}  \cline{5-5}
    \end{array}
~~& %
\tiny \hspace{-14mm}   \begin{array}{ll|l|l}
      \cline{3-3} 
          &\hspace{8mm}  & \hspace{4mm} &  \\
          \cline{2-2} 
           & ~ &\hspace{4mm} &~  \\
            &   &              &   \\
            & ~ & & \\
          \cline{2-2} 
            & ~ & & \\
                  \cline{4-4}
            & ~ & &\\
                  \cline{4-4}
            & & &\\
      \cline{3-3}  
    \end{array}
~~& %
\tiny    \begin{array}{ll|l|l}
      \cline{3-3} 
          & ~ &\hspace{4mm} & ~\\
          & ~ &\hspace{4mm} & ~\\
          &   &             &  \\
          & ~ & &\\
          & ~ & &\\
      x_6 & ~ & &\\
        \cline{2-2} \cline{4-4} 
          & & & \\
      \cline{3-3}  
    \end{array}

\\[5mm]
~~~~count~1's & ~count~2's & count~6's& ~~~~count~18's 
\end{array}
\]
\vspace{-14mm}\end{wrapfigure}
defines two
outputs $R=\tuple{r_1,r_2}$ and introduces clauses specifying that the
(unary) value of $R$ equals the (unary) value of $d_2\, \mathrm{mod~}
3$. 
\paragraph{\textbf{Step four - comparison:}}
The outputs from these four units now specify a number in base $B$,
each digit represented in unary notation. This number is now compared
(via an encoding of the lexicographic order) to $23_{(B)}$ (the value
from the right-hand side of $\psi$).

\section{Measures of Optimality}
\label{section:4}

We now return to the objective of this paper: For a given
Pseudo-Boolean constraint, how can we choose a mixed radix base with
respect to which the encoding of the constraint via sorting networks
will be optimal? We consider here three alternative cost functions
with respect to which an optimal base can be found.  These cost
functions capture with increasing degree of precision the actual size
of the encodings.

The first cost function, $\sumDigits$ as introduced in
Definition~\ref{def:sumDigits}, provides a coarse measure on the size
of the encoding.  It approximates (from below) the total number of
input bits in the network of sorting networks underlying the
encoding. An advantage in using this cost function is that there always
exists an optimal base which is prime. The disadvantage is that it
ignores the carry bits in the construction, and as such is not always
a precise measure for optimality.
In~\cite{EenS06}, the authors propose to apply a cost function which
considers also the carry bits. This is the second cost function we
consider and we call it  $\sumCarry$.

\begin{definition}[cost function: $\sumCarry$]
\label{cost2}
Let $S\in\intMultiSet$, $B\in\Base$ with $|B|=k$ and
$S_{(B)}=(a_{ij})$ the corresponding $n\times (k+1)$ matrix of digits.
Denote the sequences $\bar s=\tuple{s_0,s_1,\ldots,s_k}$ (sums) and 
$\bar c= \tuple{c_0,c_1,\ldots,c_k}$ (carries) defined by: 
$s_j=\sum_{i=1}^n a_{ij}$ for $0\leq j\leq k$, $c_0=0$, and 
$c_{j+1} = (s_j+c_j) \mbox{\rm ~div~} B(j)$ for\\
\begin{minipage}[t]{0.5\linewidth}
 $0\leq j\leq k$.
The ``sum of digits with carry''
cost function is defined by the equation on the right.
\end{minipage}\quad
\begin{minipage}[t]{0.5\linewidth}
\vspace{-6mm}
\[\sumCarry(S_{(B)}) = \sum_{j=0}^{k}(s_j+c_j)
\]
\end{minipage}
\end{definition}

The following example illustrates the $\sumCarry$ cost function and
that it provides a better measure of base optimality for the (size of
the) encoding of Pseudo-Boolean constraints.

\begin{example}\label{runningD}

  Consider the encoding of a Pseudo-Boolean constraint with
  coefficients $S=\set{1,3,4,8,18,18}$ with respect to  %
  bases: $B_1=\tuple{2,3,3}$, $B_2=\tuple{3,2,3}$, and
  $B_3=\tuple{2,2,2,2}$.
  Figure~\ref{fig:3cost} depicts the sizes of the sorting networks for
  each of these %
  bases.  The upper tables  %
  illustrate the representation of the coefficients in the
  corresponding bases. 
  In the %
  lower tables, the rows labeled ``sum'' indicate the
  number of bits per network and (to their right) their total sum which
  is the $\sumDigits$ cost.
  With respect to the $\sumDigits$ cost function, all three bases are
  optimal for $S$, with a total of 9 inputs. The algorithm might as
  well return $B_3$.

  The rows labeled ``carry'' indicate the number of carry bits in each
  of the constructions and (to their right) their totals. With respect
  to the $\sumCarry$ cost function, bases $B_1$ and $B_2$ are optimal
  for $S$, with a total of $9+2=11$ bits while $B_3$ involves $9+5=14$
  bits. The algorithm might as well return $B_1$.
\end{example}

\begin{figure}%
  \vspace{-6ex}
  \begin{center}\scriptsize  
$  
\begin{array}{|c||c|c|c|c|r}
\multicolumn{6}{c}{\qin  B_1=\tuple{2,3,3}}\\[2ex]
\cline{1-5}
  ~S~ & ~1\mbox{'s}~ & 2\mbox{'s}~ & ~6\mbox{'s} & ~18\mbox{'s} \\
\cline{1-5}
   1  &   1          &  0          &   0         &   0          \\
   3  &   1          &  1          &   0         &   0          \\
   4  &   0          &  2          &   0         &   0          \\
   8  &   0          &  1          &   1         &   0          \\
  18  &   0          &  0          &   0         &   1          \\
  18  &   0          &  0          &   0         &   1          \\
\cline{1-5}
\multicolumn{6}{c}{} \\
\cline{1-5}
sum   &   2          &  4          &   1         &   2         & ~9\\ 
carry &   0          &  1          &   1         &   0         & ~2\\
comp  &   1          &  9          &   1         &   1         & ~12\\
\cline{1-5}
\end{array}
~~~      %
\begin{array}{|c||c|c|c|c|r}
\multicolumn{6}{c}{\qin B_2=\tuple{3,2,3}}\\[2ex]
\cline{1-5}
  ~S~ & ~1\mbox{'s}~ & 3\mbox{'s}~ & ~6\mbox{'s} & ~18\mbox{'s} \\
\cline{1-5}
   1   &   1          &  0          &   0        &  0          \\
   3   &   0          &  1          &   0        &  0          \\
   4   &   1          &  1          &   0        &  0          \\
   8   &   2          &  0          &   1        &  0          \\
  18  &   0          &  0          &   0        &   1          \\
  18  &   0          &  0          &   0        &   1          \\
\cline{1-5}
\multicolumn{6}{c}{} \\
\cline{1-5}
sum   &   4          &  2          &   1         &   2         & ~9\\ 
carry &   0          &  1          &   1         &    0        & ~2\\
comp  &   5          &  3          &   1         &   1         & ~10\\
\cline{1-5}
\end{array}
~~~       %
\begin{array}{|c||c|c|c|c|c|r}
\multicolumn{7}{c}{\qin B_3=\tuple{2,2,2,2}}\\[2ex]
\cline{1-6}
~S~ & ~1\mbox{'s}~& 2\mbox{'s}~& ~4\mbox{'s}& ~8\mbox{'s}&~16\mbox{'s}\\
\cline{1-6}
 1   &   1          &  0          &   0        &   0     & 0  \\
 3   &   1          &  1          &   0        &   0     & 0  \\
 4   &   0          &  0          &   1        &   0     & 0  \\
 8   &   0          &  0          &   0        &   1     & 0  \\
18  &   0          &  1          &   0        &   0     & 1  \\
18  &   0          &  1          &   0        &   0     & 1  \\
\cline{1-6}
\multicolumn{7}{c}{} \\
\cline{1-6}
sum   &   2        &  3          &   1         &   1     & 2    & ~9\\ 
carry &   0        &  1          &   2         &   1     & 1   & ~5\\
comp  &   1        &  5          &   3         &   1     & 3   &~13\\
\cline{1-6}
\end{array}
$
\end{center}
 \caption{Number of inputs/carries/comparators when encoding
   $S=\{1,3,4,8,18,18\}$ and three bases $B_1=\tuple{2,3,3}$,
   $B_2=\tuple{3,2,3}$, and $B_3=\tuple{2,2,2,2}$ .  }
\label{fig:3cost}
\vspace{-3ex}
\end{figure}

The following example shows that  when
searching for an optimal base with respect to the $\sumCarry$ cost
function, one must consider also non-prime bases.

\begin{example}
  Consider again the Pseudo Boolean constraint $\psi = 2x_1 + 2x_2 +
  2x_3 + 2x_4 + 5x_5 + 18x_6 \geq 23$ from Section~\ref{sec:encoding}.
  The encoding with respect to $B_1=\tuple{2,3,3}$ results in 4
  sorting networks with 10 inputs from the coefficients and 2
  carries. So a total of 12 bits. The encoding with respect to
  $B_2=\tuple{2,9}$ is smaller. It has the same 10 inputs from the
  coefficients but no carry bits. Base $B_2$ is optimal and
  non-prime. 
\end{example}

We consider a third cost function which we call the $\numComparators$
cost function.
Sorting networks are constructed from ``comparators'' \cite{Knuth73}
and in the encoding each comparator is modeled using six CNF clauses.
This function counts the number of comparators in the construction.
Let $f(n)$ denote the number of comparators in an $n\times n$ sorting
network. For small values of $0\leq n\leq 8$, $f(n)$ takes the values
$0, 0, 1, 3, 5, 9, 12, 16$ and $19$ respectively which correspond to
the sizes of the optimal networks of these sizes \cite{Knuth73}.  For
larger values, the construction uses Batcher's odd-even sorting
networks \cite{Batcher68} for which
$f(n)=n\cdot\ceil{\log_2n}\cdot(\ceil{\log_2n}-1)/4+n-1$.

\begin{definition}[cost function: $\numComparators$]
\label{num_comparators}
Consider the same setting as in Definition~\ref{cost2}.  Then,
\vspace{-7mm}
\[\numComparators(S_{(B)}) = \sum_{j=0}^{k} f(s_j+c_j)\]
\end{definition}

\vspace{-5mm}
\begin{example}
  Consider again the setting of Example~\ref{runningD}. 
  In Figure~\ref{fig:3cost}  the rows labeled ``comp'' indicate the
  number of  comparators in each of the sorting networks and their
  totals. The construction with the minimal number of comparators is
  that obtained with respect to the base $B_2=\tuple{3,2,3}$ with 
  10 comparators.
\end{example}

It is interesting to remark the following relationship between the
three cost functions: The $\sumDigits$ function is the most
``abstract''. It is only based on the representation of numbers in a
mixed radix base.
The $\sumCarry$ function considers also properties of addition
in mixed-radix bases (resulting in the carry bits).
Finally, the $\numComparators$ function considers also implementation
details of the odd-even sorting networks applied in the underlying
\minisatp\ construction.
In Section~\ref{sec:exp} we evaluate how the alternative choices for a
cost function influence the size and quality of the encodings obtained
with respect to corresponding optimal bases.

\section{Optimal Base Search I: Heuristic Pruning}
\label{sec:ob1}

This section introduces a simple, heuristic-based, depth-first, tree
search algorithm to solve the optimal base problem. The search space
is the domain of non-redundant bases as presented in
Definition~\ref{def:nrb}.
The starting point is the brute force algorithm applied in \minisatp.
For a sequence of integers $S$, \minisatp\ applies a depth-first
traversal of $\Base_p^{17}(S)$ to find the base with the optimal value
for the cost function $\cost_S(B)= \sumCarry(S_{(B)})$.

Our first contribution is to introduce a heuristic function and to identify
branches in the search space which can be pruned early on in the
search. Each tree node $B$ encountered during the traversal is
inspected to check if given the best node encountered so far, $bestB$,
it is possible to determine that all descendants of $B$ are guaranteed
to be less optimal than $bestB$. In this case, the subtree rooted at
$B$ may be pruned.
The resulting algorithm improves on the one of \minisatp\ and
provides the basis for the further improvements introduced in
Sections~\ref{sec:ob2} and~\ref{sec:ob3}.
We need first a definition.

\begin{definition}[base extension, partial cost, and admissible heuristic]
  Let $S\in\intMultiSet$, $B,B'\in\Base(S)$, and $\cost_S$ a cost
  function. We say that: 
  (1) $B'$ extends $B$, denoted $B'\succ B$, if $B$ is a prefix of
  $B'$, 
  (2) $\partial \cost_S$ is a partial cost function for $\cost_S$ if
  $\forall B'\succ B. ~\cost_S(B') \geq \partial \cost_S(B)$, and
  (3) $h_S$ is an admissible heuristic function for $\cost_S$ and
  $\partial \cost_S$ if $\forall B'\succ B.  ~\cost_S(B')
  \geq \partial \cost_S(B') + h_S(B') \geq \partial \cost_S(B) +
  h_S(B) $.
\end{definition}
The intuition is that $\partial \cost_S(B)$ signifies a part of the
cost of $B$ which will be a part of the cost of any extension of $B$,
and that $h_S(B)$ is an under-approximation on the additional cost of
extending $B$ (in any way) given the partial cost of $B$.
We denote $\cost^\alpha_S(B)=\partial \cost_S(B) + h_S(B)$. If
$\partial \cost_S$ is a partial cost function and $h_S$ is an
admissible heuristic function, then $\cost^\alpha_S(B)$ is an
under-approximation of $\cost_S(B)$.
The next lemma provides the basis for heuristic pruning using the
three cost functions introduced above.

\begin{lemma}
\label{lem:h}
  The following are
    admissible heuristics for the cases when: 
  \begin{enumerate}
  \item $\cost_S(B)=\sumDigits(S_{(B)})$:~~
    $\partial \cost_S(B) = \cost_S(B)-\sum msd(S_{(B)})$.
  \item $\cost_S(B)=\sumCarry(S_{(B)})$:~~
    $\partial \cost_S(B) = \cost_S(B)-\sum msd(S_{(B)})$.

  \item $\cost_S(B)=\numComparators(S_{(B)})$:~~
   $ \partial \cost_S(B) = \cost_S(B)- f(s_{|B|}+c_{|B|})$.
  \end{enumerate}
In the first two settings we take   
$h_S(B) = \left| \sset{x \in S}{ x \geq \prod B} \right|$. \\
In the case of $\numComparators$ we take the  trivial heuristic estimate
$h_S(B)=0$
\end{lemma}

\begin{figure}[t]
\begin{SProg2}  
  {\scriptsize /*input*/} &multiset S\\
  {\scriptsize /*init*/}  &  base  bestB = $\tuple{2,2,...,2}$ \\
  {\scriptsize /*dfs*/} & depth-first traverse $\mathbf{\Base(S)}$\\
                        & at each node $B$, 
                          for the next value  p $\in$ B.extenders(S) do\\
             &\qin\qin base newB = B.extend(p)\\
             &\qin\qin if ($\cost^\alpha_S(\mbox{newB})
                               > \cost_S(\mbox{bestB})$) \textbf{prune}\\
             &\qin\qin  else if ($\cost_S(\mbox{newB})<
                                              \cost_S(\mbox{bestB})$) 
                                           bestB = newB \\                  
   {\scriptsize /*output*/} & return bestB;\\
\end{SProg2}
\caption{\texttt{dfsHP}: depth-first search for an
  optimal base with heuristic pruning}
\label{fig:alg1}
\vspace{-3ex}
\end{figure}

The algorithm, which we call \texttt{dfsHP} for depth-first search
with heuristic pruning, is now stated as Figure~\ref{fig:alg1} where
the input to the algorithm is a \texttt{multiset} of integers
\texttt{S} and the output is an optimal base. The algorithm applies a
depth-first traversal of $\Base(S)$ in search of an optimal base.
We assume given: a cost function $\cost_S$, a partial cost function
$\partial \cost_S$ and an admissible heuristic $h_S$. We denote
$\cost^\alpha_S(B)=\partial \cost_S(B) + h_S(B)$.
The abstract data type \texttt{base} has two operations:
\texttt{extend(int)} and \texttt{extenders(multiset)}. For a base
\texttt{B} and an integer \texttt{p}, \texttt{B.extend(p)} is the base
obtained by extending \texttt{B} by \texttt{p}. For a multiset
\texttt{S}, \texttt{B.extenders(S)} is the set of integer values
\texttt{p} by which \texttt{B} can be extended to a non-redundant base
for \texttt{S}, i.e., such that $\prod \mbox{\texttt{B.extend(p)}}\leq
max(\mbox{\texttt{S}})$.  The definition of this operation may have
additional arguments to indicate if we seek a prime base or one
containing elements no larger than $\ell$.

Initialization ({\scriptsize /*init*/} in the figure) assigns to the
variable \texttt{bestB} a finite binary base of size
$\lfloor\log_2(max(S))\rfloor$. This variable will always denote the
best base encountered so far (or the initial finite binary base).
Throughout the traversal, when visiting a node \texttt{newB} we first
check if the subtree rooted at \texttt{newB} should be pruned. If this
is not the case, then we check if a better ``best base so far'' has
been found. Once the entire (with pruning) search space has been
considered, the optimal base is in \texttt{bestB}.

To establish a bound on the complexity of the algorithm,  denote
the number of different integers in $S$ by $s$ and $m=max(S)$.  
The algorithm has space complexity $O(\log(m))$, for the depth first
search on a tree with height bound by $\log(m)$ (an element of
$\Base(S)$ will have at most $\log_2(m)$ elements).
For each base considered during the traversal, we have to calculate
$cost_S$ which incurs a cost of $O(s)$.  To see why, consider that
when extending a base $B$ by a new element giving base $B'$, the first
columns of $S_{(B')}$ are the same as those in $S_{(B)}$ (and thus
also the costs incurred by them). Only the cost incurred by the most
significant digit column of $S_{(B)}$ needs to be recomputed for
$S_{(B')}$ due to base extension of $B$ to $B'$. 
Performing the
computation for this column, we compute a new digit for the $s$
different values in $S$.
Finally, by Lemma~\ref{zeta}, there are $O(m^{2.73})$ bases and
therefore, the total runtime is $O(s*m^{2.73})$. Given that $s\leq m$,
we can conclude that runtime is bounded by $O(m^{3.73})$.
\section{Optimal Base Search II: Branch and Bound}
\label{sec:ob2}

In this section we further improve the search algorithm for an optimal
base. 
The search algorithm is, as before, a traversal of the search space
using the same partial cost and heuristic functions as before to prune
the tree. The difference is that instead of a depth first search, we
maintain a priority queue of nodes for expansion and apply a
best-first, branch and bound search strategy.

Figure~\ref{fig:alg2} illustrates our enhanced search algorithm. We
call it \texttt{B\&B}.
The abstract data type \texttt{priority\_queue} maintains bases
prioritized by the value of $\cost^\alpha_S$. Operations
\texttt{popMin()}, \texttt{push(base)} and \texttt{peek()} (peeks at
the minimal entry) are the usual. The reason to box the text
``\texttt{priority\_queue}'' in the figure will become apparent in
the next section.

\begin{figure}[t]
  \begin{SProg}
  base findBase(multiset S)\\
  {\scriptsize /*1*/} \qin base  bestB = $\tuple{2,2,...,2}$; 
                           \fbox{priority\_queue} Q = $\set{\tuple{~}}$;\\
  {\scriptsize /*2*/} \qin while ($\mbox{Q} \neq \set{}$ \mbox{\&\&} 
        $\cost^\alpha_S(\mbox{Q.peek()}) < \cost_S(\mbox{bestB})$) \\
  {\scriptsize /*3*/} \qin\qin  base B = Q.popMin();\\
  {\scriptsize /*4*/} \qin\qin  foreach (p $\in$ B.extenders(S)) \\
  {\scriptsize /*5*/} \qin\qin\qin base newB = B.extend(p);\\
  {\scriptsize /*6*/} \qin\qin\qin  if ($\cost^\alpha_S(\mbox{newB})
                                         \leq \cost_S(\mbox{bestB})$) \\
  {\scriptsize /*7*/} \qin\qin\qin\qin  \{ Q.push(newB); 			
        if ($\cost_S(\mbox{newB})<\cost_S(\mbox{bestB})$) bestB =
        newB; \}\\ 
  {\scriptsize /*8*/} \qin return bestB;\\
  \end{SProg}
\caption{Algorithm \texttt{B\&B}: best-first, branch and bound}
\label{fig:alg2}
\vspace{-3ex}
\end{figure}

On line {\scriptsize /*1*/} in the figure, we initialize the variable
\texttt{bestB} to a finite binary base of size
$\lfloor\log_2(max(S))\rfloor$ (same as in Figure~\ref{fig:alg1})
and initialize the queue to contain the
root of the search space (the empty base).
As long as there are still nodes to be expanded in the queue that are
potentially interesting (line {\scriptsize /*2*/}), we select (at line
{\scriptsize /*3*/}) the best candidate base \texttt{B} from the
frontier of the tree in construction for further expansion.  Now the
search tree is expanded for each of the relevant integers (calculated
at line {\scriptsize /*4*/}). For each child \texttt{newB} of
\texttt{B} (line {\scriptsize /*5*/}), we check if pruning at
\texttt{newB} should occur (line {\scriptsize /*6*/}) and if not we
check if a better bound has been found (line {\scriptsize /*7*/})
Finally, when the loop terminates, we have found the optimal base and
return it (line {\scriptsize /*8*/}).

\section{Optimal Base Search III: Search Modulo Product}
\label{sec:ob3}

This section introduces an abstraction on the search space,
classifying bases according to their product. Instead of maintaining
(during the search) a priority queue of all bases
(nodes) that still need to be explored, we maintain a special priority
queue in which there will only ever be at most one base with the same
product.  So, the queue will never contain two
different
bases $B_1$ and $B_2$
such that $\prod B_1 = \prod B_2$.
In case a second base, with the same product as one already in, is
inserted to the queue, then
only
the base with the minimal value of
$\cost^\alpha_S$ is maintained on the queue. We call this type of
priority queue a \emph{hashed priority queue} because it
can conveniently be
implemented as a hash table.

The intuition comes from a study of the $\sumDigits$ cost function for
which we can prove the following \textbf{Property 1} on bases:
Consider two bases $B_1$ and $B_2$ such that $\prod B_1 = \prod B_2$
and such that $\cost^\alpha_S(B_1) \leq \cost^\alpha_S(B_2)$. Then for
any extension of $B_1$ and of $B_2$ by the same sequence $C$,
$\cost_S(B_1C) \leq \cost_S(B_2C)$.  In particular, if one of $B_1$ or
$B_2$ can be extended to an optimal base, then $B_1$ can.
A direct implication is that when maintaining the frontier of the
search space as a priority queue, we only need one representative of
the class of bases which have the same product (the one with the
minimal value of $\cost^\alpha_S$).

A second \textbf{Property 2} is more subtle and true for any cost
function that has the first property: Assume that in the algorithm
described as Figure~\ref{fig:alg2} we at some stage remove a base
$B_1$ from the priority queue. This implies that if in the future we
encounter any base $B_2$ such that $\prod B_1 = \prod B_2$, then we
can be sure that $\cost_S(B_1) \leq \cost_S(B_2)$ and immediately
prune the search tree from $B_2$.

Our third and final algorithm, which we call \texttt{hashB\&B}
(best-first, branch and bound, with hash priority queue) is identical
to the algorithm presented in Figure~\ref{fig:alg2}, except that the
the boxed priority queue introduced at line {\scriptsize /*1*/} is
replaced by a \begin{tt}\fbox{hash\_priority\_queue}\end{tt}.

The abstract data type \texttt{hash\_priority\_queue} maintains bases
prioritized by the value of $\cost^\alpha_S$. Operations
\texttt{popMin()} and \texttt{peek()} are as usual. Operation
\texttt{push($B_1$)} works as follows: (a) if there is no base $B_2$
in the queue such that $\prod B_1=\prod B_2$, then add $B_1$.
Otherwise, (b) if $\cost^\alpha_S(B_2)\leq \cost^\alpha_S(B_1)$ then
do not add $B_1$. Otherwise, (c) remove $B_2$ from the queue and add
$B_1$.
\begin{theorem}\label{algIsGood}~\\
\noindent(1) The $\sumDigits$ cost function satisfies \textbf{Property 1}; and
(2)
the \texttt{hashB\&B} algorithm finds an optimal base for any cost
function which satisfies \textbf{Property~1}.
\end{theorem}

We conjecture that the other cost functions do not satisfy
\textbf{Property 1}, and hence cannot guarantee that the
\texttt{hashB\&B} algorithm always finds an optimal base.
However, in our extensive experimentation, all bases found
(when searching for an optimal prime base) are indeed optimal.

A direct implication of the above improvements is that we can now
provide a tighter bound on the complexity of the search algorithm.
Let us denote the number of different integers in $S$ by $s$ and
$m=max(S)$.
First note that in the worst case the hashed priority queue will
contain $m$ elements (one for each possible value of a base product,
which is never more than $m$).  
Assuming that we use a Fibonacci Heap, we have a $O(\log(m))$ cost
(amortized) per \texttt{popMin()} operation and in total a $O(m*
\log(m))$ cost for popping elements off the queue during the search
for an optimal base.

Now focus on the cost of operations performed when extracting a
base $B$ from the queue. Denoting
$\prod B=q$,
$B$ has at most $m/q$
children (integers which extend it). 
For each child we have to calculate $cost_S$ which incurs a cost of
$O(s)$ 
and possibly to insert it to the queue.
Pushing an element onto a hashed priority queue (in all three cases)
is a constant time operation (amortized), and hence the total cost for
dealing with a child is $O(s)$.

Finally, consider the total number of children created during
the search which corresponds to the following sum:
\vspace{-2mm}
\[  
  O(\sum_{q=1}^{m}m/q ) = O(m\sum_{q=1}^{m}1/q)=O(m*\log(m)) 
\] 
So, in total we get $O(m*\log(m))+O(m*\log(m)*s)\leq O(m^2 *
\log(m))$. 
When we restrict the extenders to be prime numbers then we can further
improve this bound to $O(m^2 * \log(\log(m)))$ by reasoning about the
density of the primes. 
A proof can be found in the appendix.
\section{Experiments}
\label{sec:exp}

Experiments are performed using an extension to \minisatp\
\cite{EenS06} where the only change to the tool is to plug in our
optimal base algorithm.
The reader is invited to experiment with the implementation via its
web interface.\footnote{
  \url{http://aprove.informatik.rwth-aachen.de/forms/unified_form_PBB.asp}}
All experiments are performed on a Quad-Opteron 848 at 2.2 GHz, 16 GB
RAM, running Linux.

Our benchmark suite originates from 1945 Pseudo-Boolean Evaluation
\cite{Manquinho06}
instances from the years 2006--2009 containing a total of 
74,442,661 individual Pseudo-Boolean constraints. After 
normalizing and removing constraints with $\set{0,1}$ coefficients we
are left with 115,891 different optimal base problems where the
maximal coefficient is $2^{31}-1$.
We then focus on 734
PB instances where at least one optimal base
problem from the instance yields a base with an element that is
non-prime or greater than 17.
When solving PB instances, in all experiments, a 30 minute timeout is
imposed as in the Pseudo-Boolean competitions. When solving an optimal
base problem, a 10 minute timeout is applied.
\paragraph{Experiment 1 (Impact of optimal bases):}

The first experiment illustrates the advantage in searching for an
optimal base for Pseudo-Boolean solving. We compare sizes and
solving times when encoding w.r.t. the binary base vs.\ w.r.t. an optimal base
(using the \texttt{hashB\&B} algorithm with the $\numComparators$ cost
function).
Encoding w.r.t. the binary base, we solve 435 PB instances (within the
time limit) with an average time of 146 seconds and average CNF size
of 1.2 million clauses.
Using an optimal base, we solve 445 instances with an average time of
108 seconds, and average CNF size of 0.8 million clauses.

\paragraph{Experiment 2 (Base search time):}

Here we focus on the search time for an optimal base in six
configurations using the $\sumCarry$ cost function.
Configurations \texttt{M17}, \texttt{dfsHP17}, and \texttt{B\&B17},
are respectively, the \minisatp\ implementation, our \texttt{dfsHP}
and our \texttt{B\&B} algorithms, all three searching for an optimal
base from $\Base_p^{17}$, i.e., with prime elements up to 17.
Configurations \texttt{hashB\&B1,000,000}, \texttt{hashB\&B10,000},
and \texttt{hashB\&B17} are our \texttt{hashB\&B} algorithm 
searching for a base from $\Base_p^\ell$ with bounds of
$\ell = $ 1,000,000, $\ell =$ 10,000, and $\ell =$ 17, respectively.

Results are summarized in Fig.~\ref{fig:results1} which is obtained
as follows. We cluster the optimal base problems according to the
values $\lceil \log_{1.9745} M \rceil$ where $M$ is the maximal
coefficient in a problem. Then, for each cluster we take the average
runtime for the problems in the cluster. The value $1.9745$ is chosen
to minimize the standard deviation from the averages (over all
clusters). These are the points on the graphs.
Configuration \texttt{M17} times out on 28
problems.  For \texttt{dfsHP17}, the maximal search time is 200
seconds.  Configuration \texttt{B\&B17} times out for 1 problem.  The
\texttt{hashB\&B} configurations have maximal runtimes of 350 seconds,
14 seconds and 0.16 seconds, respectively for the bounds 1,000,000,
10,000 and 17.
\begin{figure}
  \centering
  \vspace{-3ex}
  \includegraphics[scale=0.152]{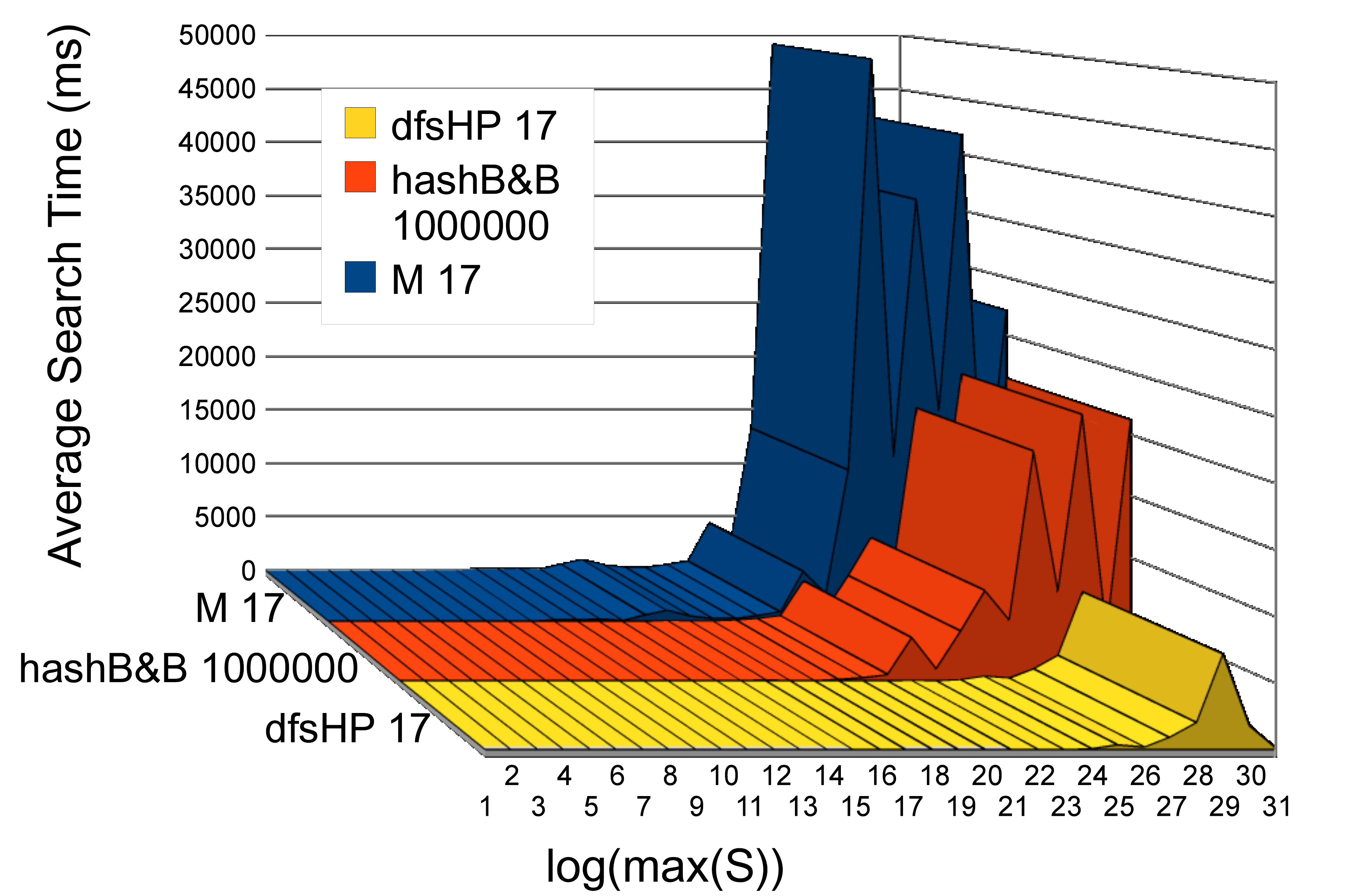}
  \includegraphics[scale=0.152]{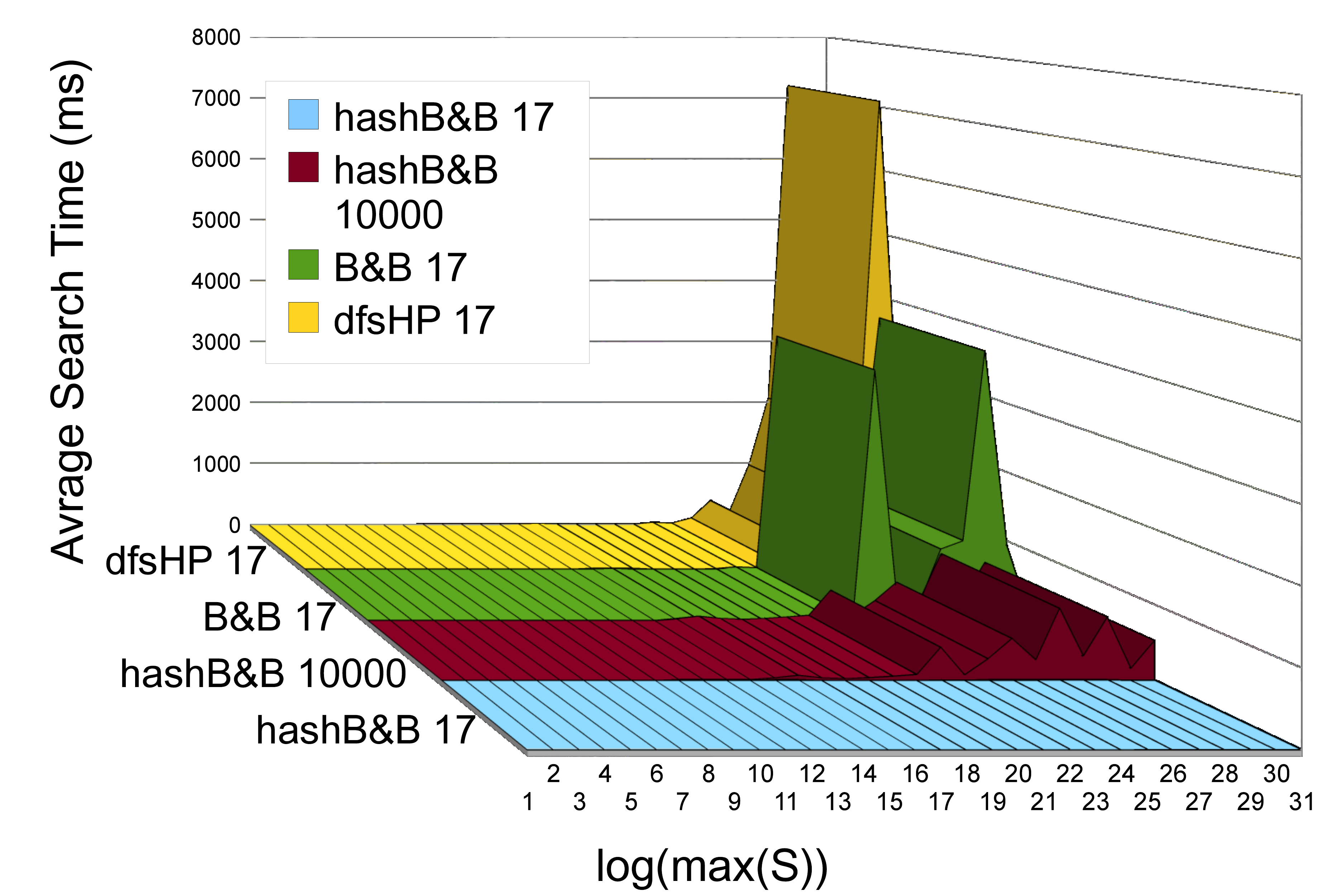}
  \vspace{-3ex}
  \caption{Experiment 2: the 3 slowest configurations (left)
    (from back to front)
    \texttt{M17}(blue), \texttt{hashB\&B1,000,000}(orange) and
    \texttt{dfsHP17}(yellow).  
    The 4 fastest configurations (right) (from back to front)
    \texttt{dfsHP17}(yellow), \texttt{B\&B17}(green),
    \texttt{hashB\&B10,000}(brown) and \texttt{hashB\&B17}(azure).
    Note the change of scale for the $y$-axis with 50k ms on the left
    and 8k ms on the right. Configuration \texttt{dfsHP17} (yellow) is
    lowest on left and highest on right, setting the reference point
    to compare the two graphs. \vspace{-5mm}}

\label{fig:results1}
\vspace{-1ex}
\end{figure}

Fig.~\ref{fig:results1} shows  that: (left)  even with primes
up to 1,000,000, \texttt{hashB\&B} is faster than the algorithm from
\minisatp\ with the limit of 17; and
(right)  even with primes
up to 10,000, the search time using \texttt{hashB\&B} is essentially
negligible.
\paragraph{Experiment 3 (Impact on PB solving):}
Fig.~\ref{fig:results2} illustrates the influence of improved base
search on SAT solving for PB Constraints. Both graphs
depict the number of instances solved (the $x$-axis) within a
time limit (the $y$-axis). On the left, total solving time
(with base search), and on the right, SAT solving time only.

\begin{figure}
  \centering
  \vspace{-5ex}
  \includegraphics[scale=0.063]{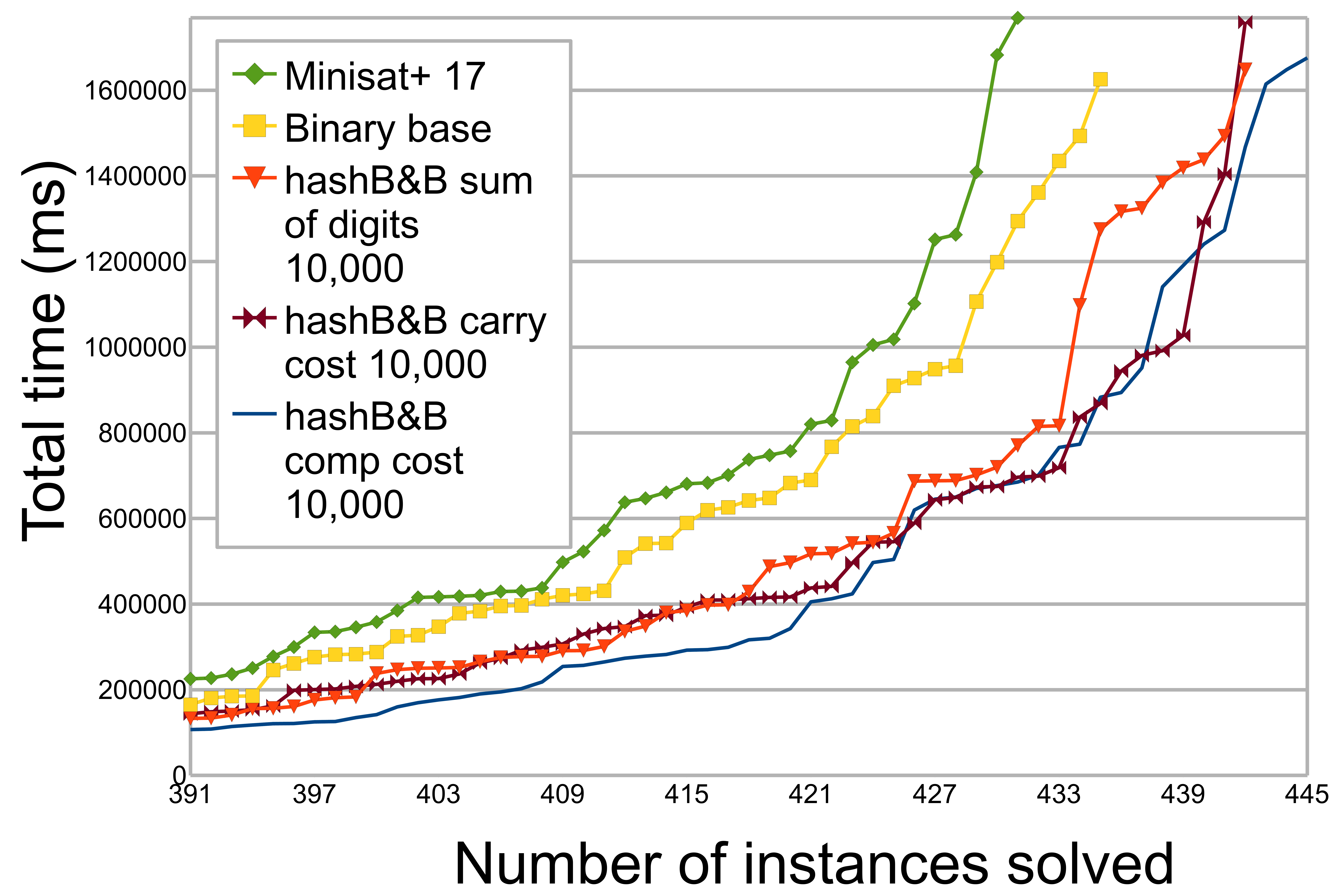}
  \includegraphics[scale=0.063]{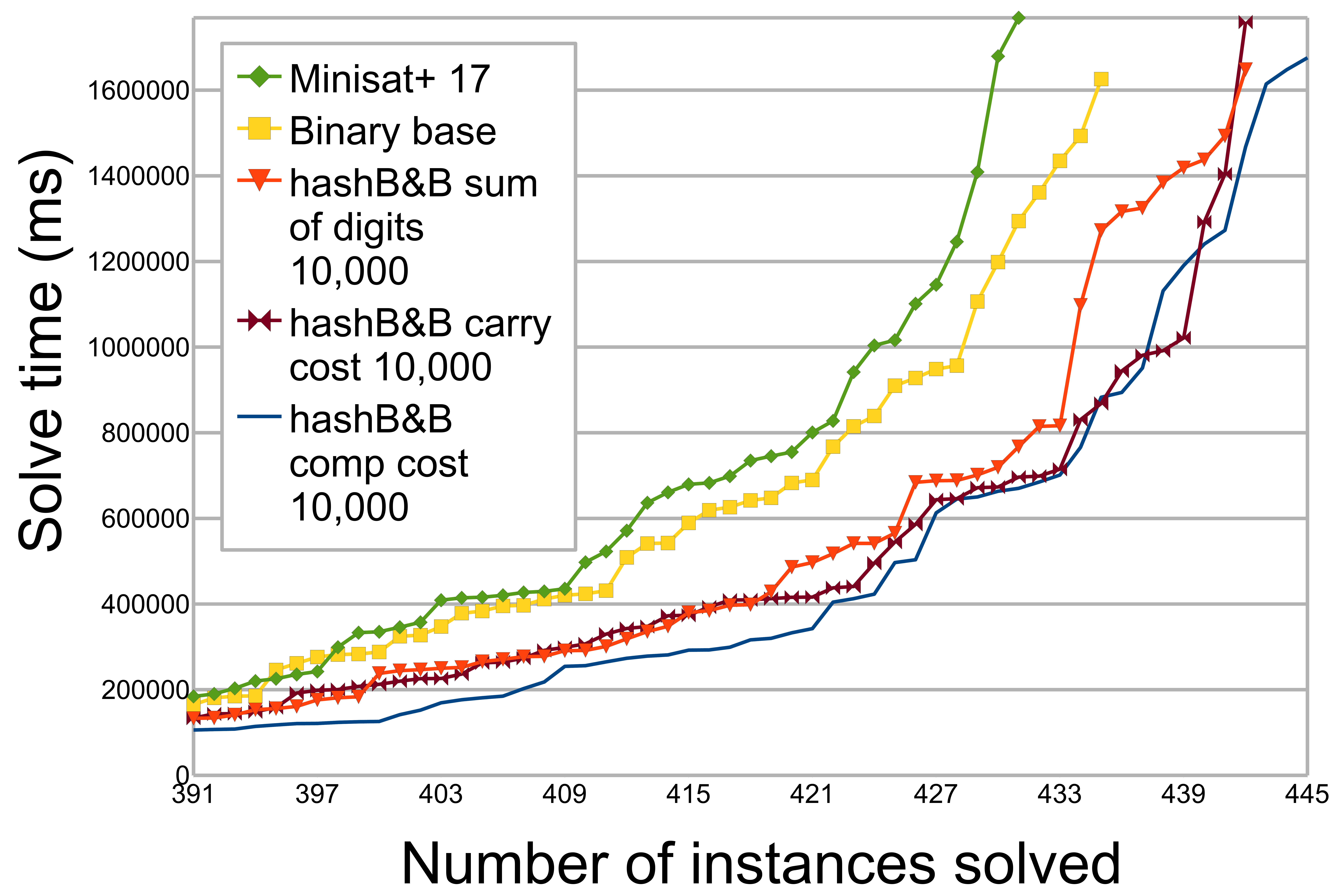}
  \vspace{-3ex}
  \caption{Experiment 3: total times (left), solving times (right) %
 }
 \vspace{-5ex}
\label{fig:results2}
\end{figure}
\noindent Both graphs consider the $734$ instances of interest and compare
SAT solving times with bases found using five configurations.  %
The first is \minisatp with
configuration \texttt{M17}, the second is with respect to the binary
base, the third to fifth are \texttt{hashB\&B}
searching for bases from $\Base_p^{\text{10,000}}(S)$
with cost functions: $\sumDigits$, $\sumCarry$, and
$\numComparators$, respectively.
The average total/solve run-times (in sec) are 150/140, 146/146,
122/121, 116/115 and 108/107 (left to right). The total number of
instances solved are 431, 435, 442, 442 and 445 (left to right).  
The average CNF sizes (in millions of clauses) for %
the entire
test set/%
the set where all algorithms solved/%
the set where
no algorithm solved are 7.7/1.0/18,
9.5/1.2/23, 8.4/1.1/20, 7.2/0.8/17
and 7.2/0.8/17 (left to right).  

The graphs of Fig.~\ref{fig:results2} and average solving times
clearly show: \textbf{(1)} SAT solving time dominates base
finding time, \textbf{(2)} \minisatp\ is outperformed by the trivial
binary base,  \textbf{(3)}
total solving times with our algorithms are faster than with the binary
base, and \textbf{(4)} the most specific
cost function (comparator cost) outperforms the other cost functions
both in solving time and size.  Finally, note that sum of
digits with its nice mathematical properties, simplicity, and
application independence solves as many instances as cost carry.

\paragraph{Experiment 4 (Impact of high prime factors):}

\begin{wrapfigure}{r}{73mm}\vspace{0mm}
  \includegraphics[scale=0.16]{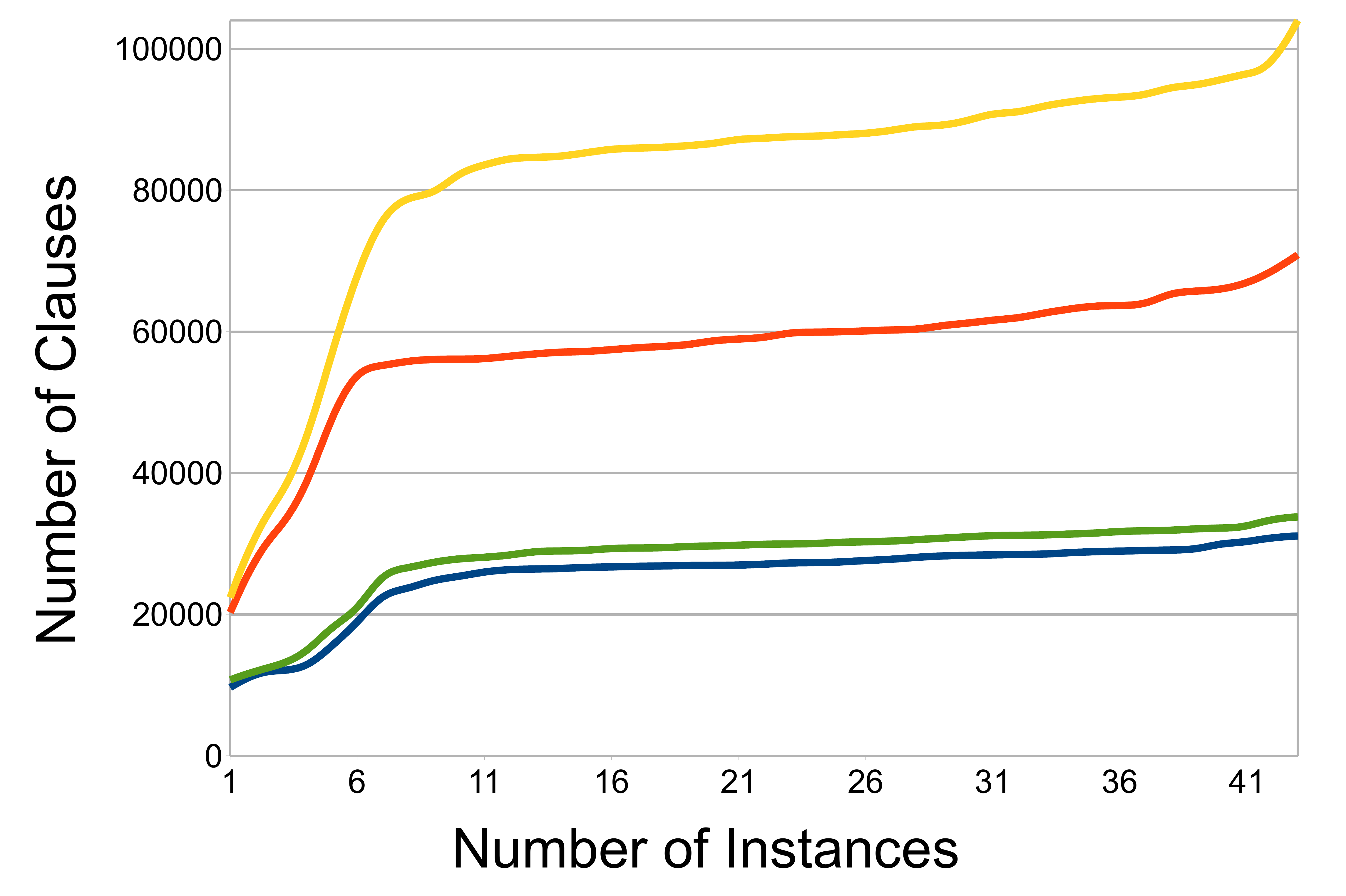}
  \vspace{-4ex}
  \caption{\small
    Experiment 4:  Number ($x$-axis) of instances encoded within
    number of clauses ($y$-axis) on 4 configurations. From top line
    to bottom: (yellow) $\ell=17$, $i=5$, (red) $\ell=17$, $i=2$,
    (green) $\ell=31$, $i=5$, and (blue) $\ell\in\{17,31\}$, $i=0$.}
\label{fig:resultsGenCactus}
\vspace{-4ex}
\end{wrapfigure}
This experiment is about the effects of restricting the maximal prime
value in a base (i.e. the value $\ell=17$ of \minisatp).
An analysis of the our benchmark suite indicates that coefficients
with small prime factors are overrepresented.
To introduce instances where coefficients have larger prime factors we
select 43 instances from the suite and multiply their coefficients to
introduce the prime factor 31 raised to the power $i \in
\{0,\ldots,5\}$. We also introduce a slack
variable 
to avoid gcd-based simplification.
This gives us a collection of 258 new instances.
We used the \texttt{B\&B} algorithm
with the $\sumCarry$ cost function applying the limit $\ell=17$ (as in
\minisatp) and $\ell=31$.
Results indicate that for $\ell=31$, both CNF size
and
SAT-solving time are independent of the factor $31^i$ introduced for
$i>0$. However, for $\ell=17$, both measures increase as the power
$i$ increases. Results on CNF sizes are reported in Fig.\
\ref{fig:resultsGenCactus} which plots for 4 different settings the
number of instances encoded ($x$-axis) within a CNF with that many
clauses ($y$-axis).

\section{Related Work}\label{relwork}

Recent work \cite{BailleuxBR09} encodes Pseudo-Boolean constraints via
``totalizers'' similar to sorting networks, determined by the
representation of the coefficients in an underlying base.  Here the
authors choose the standard base 2 representation of numbers. It is
straightforward to generalize their approach for an arbitrary mixed
base, and our algorithm is directly applicable.
In~\cite{Sidorov99} the author considers the $\sumDigits$ cost
function and analyzes the size of representing the natural numbers up
to $n$ with (a particular class of) mixed radix bases.  
Our Lemma~\ref{lem:primes} may lead to a contribution in that context.

\section{Conclusion}
\label{sec:conc}

It has been recognized now for some years that decomposing the
coefficients in a Pseudo-Boolean constraint with respect to a mixed
radix base can lead to smaller SAT encodings. However, it remained an
open problem to determine if it is feasible to find such an optimal
base for constraints with large coefficients. In lack of a better
solution, the implementation in the \minisatp\ tool applies a brute
force search considering prime base elements less than 17.

To close this open problem, we first formalize the optimal base
problem and then significantly improve the search algorithm currently
applied in \minisatp.
Our algorithm scales and easily finds optimal bases
with elements up to 1,000,000.
We also illustrate that, for the
measure of optimality applied in \minisatp, one must consider also
non-prime base elements. However, choosing the more simple
$\sumDigits$ measure, it is sufficient to restrict the search to prime
bases.

With the implementation of our search algorithm it is possible, for
the first time, to study the influence of basing SAT encodings on
optimal bases.  We show that for a wide range of benchmarks,
\minisatp\ does actually find an optimal base consisting of elements
less than 17.  We also show that many Pseudo-Boolean instances have
optimal bases with larger elements and that this does influence the
subsequent CNF sizes and SAT solving times, especially when
coefficients contain larger prime factors.

\vspace{-2ex}

\subsubsection*{Acknowledgement}
We thank Daniel Berend and Carmel Domshlak for useful discussions.

\vspace{-2ex}

\newcommand{\noopsort}[1]{}

\appendix

\section{Appendix: Proofs}\label{apdx}

\newcommand{\apdxtheoremlike}[1]{\par\medskip\penalty-250{\bfseries\scshape\noindent#1}\slshape}

\subsection{Proving Lemma~\ref{l1}}

\noindent\apdxtheoremlike{Lemma~\ref{l1}.}
Let $S \in \intMultiSet$ and consider the $\sumDigits$ cost
function. Then, $S$ has an optimal base in $\Base(S)$.

\begin{proof}
  Let $S\in\intMultiSet$ and let $B\in\Base$ be an optimal base for
  $S$ with $\prod B > max(S)$. Let $B'$ be the base obtained by
  removing the last element from $B$.  We show that $msd(S_{(B)}) =
  \tuple{0, \ldots, 0}$ and that $\sumDigits(S_{(B)}) =
  \sumDigits(S_{(B')})$. The claim then follows.
  Assume falsely that for $v \in S$,\linebreak $msd(v_{(B)})>0$.
  Then we get the contradiction $v =
  \prod_{i=0}^{|B|}v_{(B)}(i) \cdot weights(B)(i) \geq msd(v_{(B)})
  \cdot \prod B > msd(v_{(B)}) \cdot max(S) \geq v$.
  From the definition of $\sumDigits$ and the above
  contradiction we deduce that $\sumDigits(B)= \sumDigits(B')$.
\end{proof}

It is straightforward to generalize Lemma~\ref{l1} for the other cost
functions considered in this paper. 

\subsection{Proving Lemma~\ref{lem:primes}} 

\begin{proposition}[Unique base representation] \label{pNib} %
  Let $B=\tuple{r_0,\ldots,r_{k-1}}$ be a base with
  $weights(B)=\tuple{w_0,\ldots,w_k}$ and let $v\in\mathbb{N}$.  
  Then, the unique representation of $v$ in $B$ is obtained as
  $v_{(B)}=\tuple{d_0,\ldots,d_k}$ such that:
(1) $d_0 = mod(v,r_0)$,             
(2) $d_i = mod( div(v,w_i) , r_i)$ \ for $0<i<k$, and
(3) $d_k = div(v,w_k)$.
\end{proposition}

\begin{proof}(sketch)
  We have to show that for $i<k$, $0\leq d_i<r_i$ and that 
  $v=\sum_{i=0}^{k} d_i\cdot w_i$. The first property follows directly
  from the construction. The second is elaborated on below.
\end{proof}

\begin{proposition}[Base factoring]
\label{pPrimesOnly}
Let $v \in \mathbb{N}$ and let $B_1 = \tuple{r_0,\ldots,r_{k-1}}$ and
$B_2 = \tuple{r'_0,\ldots,r'_{k-2}}$ be bases which are identical
except that at some position $0\leq p<k-1$, two consecutive base
elements in $B_1$ are replaced in $B_2$ by their
multiplication. Formally:
 $r'_i=r_i$ for $0\leq i < p$, 
 $r'_{p}=r_p \cdot r_{p+1}$, and
 $r'_i=r_{i+1}$ for $p<i<k-1$.
Then, the sum of the digits in $v_{(B_2)}$ is greater or equal to the
sum of the digits in $v_{(B_1)}$.
\end{proposition}

\begin{proof}(sketch)
   We first observe that 
	$v_{(B_2)}(i)=v_{(B_1)}(i)$  for $0 \leq i < p$   and 
        $v_{(B_2)}(i)=v_{(B_1)}(i+1)$ for $p< i<k-1$.
   So, it remains to show that 
	$ \sum_{i=0}^{k-1}v_{(B_2)}(i) - \sum_{i=0}^{k}v_{(B_1)}(i)= 
                          v_{(B_2)}(p) -v_{(B_1)}(p)-v_{(B_1)}(p+1) \geq 0$.
   We elaborate on this below.
\end{proof}

\noindent\apdxtheoremlike{Lemma~\ref{lem:primes}.}
   Let $S \in \intMultiSet$ and consider the $\sumDigits$ cost
   function. Then, $S$ has an optimal base in $\Base_p(S)$.

\begin{proof}
  Let $B_2$ be a base with $k-1$ elements of the form
\[
    B_2 = \tuple{r_1,\ldots,r_{p-1},(r_p\cdot r_{p+1}),
                  r_{p+2},\ldots,r_{k-1}}
\]
  where the element $(r_p\cdot r_{p+1})$ at position $p$
  is non-prime and $r_p, r_{p+1} > 1$.
  So, taking $B_1= \tuple{r_1,\ldots,r_{p-1},r_p,r_{p+1},
                  r_{p+2},\ldots,r_{k-1}}$, we are in the  setting of
  Proposition~\ref{pPrimesOnly}.
  The result follows:
  \begin{eqnarray*}
  \sumDigits(S_{(B_2)}) - \sumDigits(S_{(B_1)}) &=& 
     \sum_{v \in S}  \sum_{i=0}^{k-1}v_{(B_2)}(i)- \sum_{v \in S}
     \sum_{i=0}^{k}v_{(B_1)}(i)\\ 
  &=& \sum_{v \in S} ( v_{(B_2)}(p) - v_{(B_1)}(p)- v_{(B_1)}(p+1) )\\
  &\geq& 0    
  \end{eqnarray*}
\end{proof}

\subsection{Proving Lemma~\ref{lem:h}}
Consider the
notation of Definition~\ref{cost2}.  Let $S\in\intMultiSet$,
$B\in\Base$ with $|B|=k$ and $S_{(B)}=(a_{ij})$.
Denote the sequences $\bar s(B)=\tuple{s_0^B,s_1^B,\ldots,s_k^B}$
(sums) and $\bar c(B)= \tuple{c_0^B,c_1^B,\ldots,c_k^B}$ (carries)
defined by:
$s_j^B=\sum_{i=1}^n a_{ij}$ for $0\leq j\leq k$, $c_0^B=0$, and 
$c_{j+1}^B = (s_j^B+c_j^B) \mbox{\rm ~div~} B(j)$ for $0\leq j\leq k$.
We denote also $inputs_S^B(j)=s_{j-1}^B+c_{j-1}^B$.

\begin{proposition}
\label{pSubBase1} 
Let $S\in\intMultiSet$ and $B$,$B'$ bases such that $B'\succ B$.
Then, for $1 \leq j \leq |B|$, $inputs_S^B(j)=inputs_S^{B'}(j)$.
\end{proposition}

\begin{proposition}
\label{pSubBase2}
Let $0 \neq v\in \mathbb{N}$, and $B$ a base.  Then for every $0 \leq
i \leq |B|$ such that $v \geq weights(B)(i)$ there exists $i\leq j \leq
|B|$ such that $v_{(B)}(j) > 0$.
\end{proposition}

\begin{proposition}
\label{lHC}
Let $f:\mathbb{R} \mapsto \mathbb{R}$ be any monotonically increasing
function such that for $x,y \in \mathbb{R_{+}}$, $f(x+y) \geq f(x)+y$.
Let $S\in\intMultiSet$. Then $\tuple{cost_S^f, \partial
  \cost_S^f,h_S^f}$ given by:
\begin{enumerate}
\item 
$cost_S^f(B) = \sum_{j=1}^{|B|+1} f(inputs_S^B(j))$
\item 
$\partial\cost_S^f(B) =
	\left(\sum_{j=1}^{|B|} 
           f(inputs_S^B(j))
        \right) + f(div(inputs_S^B(|B|),B(|B|-1)))$.
\item 
$h_S^f(B) = \left| \sset{x \in S}{ x \geq \prod B} \right|$.
\end{enumerate}
	is a heuristic cost function.
\end{proposition}

\begin{proof} 
  Let $f$ and $S$ be as defined above and $B$, $B'$ bases such that
  $B'\succ B$, $|B|=k$, and $|B'|=k'$.  For every base $B''$ we can
  see that $h_S^f(B'')\geq 0$ (the size of a finite set is a natural
  number).  Therefore it is sufficient to prove that $\cost_S(B')
  \geq \partial \cost_S(B) + h_S(B)$.  From
  Proposition~\ref{pSubBase1} we get that
  $inputs_S^B(k)=inputs_S^{B'}(k)$.  From the definition of $inputs$
  and of $f$ we see that
  \begin{eqnarray*}
     f(inputs_S^{B'}(k+1)) &=& 
         f(\sum_{v\in S}v_{(B')}(k)+ div(inputs_S^{B'}(k),B'(k-1)))\\
     &=&  f(\sum_{v\in S}v_{(B')}(k)+ div(inputs_S^{B}(k),B(k-1)))\\ 
     &\geq& \sum_{v\in S}v_{(B')}(k)+ f(div(inputs_S^{B}(k),B(k-1))) \\
  \end{eqnarray*}
  For $k \leq j \leq k'$, denote $nz(j)=\sset{v \in
    S}{v_{(B')}(j)>0}$.  \\
  Let $v \in \sset{x \in S}{ x \geq \prod B =
    weights(B')(k)}$. By Proposition \ref{pSubBase2} there exists $k
  \leq j \leq k'$ such that $v_{(B')}(j) > 0 $ and therefore $v \in
  nz(j)$.  This implies that
  \begin{eqnarray*}
    | \sset{x \in S}{ x \geq \prod B } |  +  
         f(div(inputs_S^{B}(k),B(k-1))) &\leq& \\
    \sum_{j=k}^{k'} \sum_{v \in nz(j)} v_{(B')}(j)  + 
          f( div(inputs_S^{B}(k),B(k-1)))  & =& \\
    \sum_{j=k}^{k'} \sum_{v \in S} v_{(B')}(j)  +  
        f(div(inputs_S^{B}(k),B(k-1)))  & \leq& 
    \sum_{j=k+1}^{k'+1}f( inputs_S^{B'}(j))
 \end{eqnarray*}
 In total we have  
\begin{eqnarray*}
   \lefteqn{\cost_S(B')  - \partial \cost_S(B) - h_S(B) =}\\
&&     \sum_{j=|B|+1}^{|B'|+1} f(inputs_S^{B'}(j)) - 
                 f(div(inputs_S^B(|B|),B(|B|-1)))-  \\
&&               
               \quad\left| \sset{x \in S}{ x \geq \prod B} \right| \geq 0
\end{eqnarray*}
The proof that $\partial \cost_S(B') + h_S(B') \geq \partial
\cost_S(B) + h_S(B)$ is similar noting that $|nz(k')|=h_S(B')$.
\end{proof}

\begin{proof}(of Lemma~\ref{lem:h})\\
  If $f(x)=x$, then $cost_S^f =
  \sumCarry$.  By Proposition~\ref{lHC} both
  definitions give a heuristic cost function. 
  The proof for the case of $\numComparators$ and $\sumDigits$ 
   is of the same structure as Proposition~\ref{lHC}. The case of $ \sumCarry$ is the most complicated one.
\end{proof}

\subsection{Proving Theorem~\ref{algIsGood}}

\begin{definition}(\textbf{Property 1})
\label{dBME}
A heuristic cost function $\tuple{\cost_S,\partial \cost_S,h_S}$ is
called \emph{base mul equivalent} if for every $S\in\intMultiSet$ and
for every bases $B_1$, $B_2$ such that $\prod B_1 = \prod B_2$ and
$\cost^\alpha_S(B_1) \leq \cost^\alpha_S(B_2)$ the following holds:
\begin{enumerate}
   \item for any extension of $B_1$ and of $B_2$ by same base $C$,
          $\cost_S(B_1C) \leq \cost_S(B_2C)$.
\end{enumerate}
\end{definition} 

In the following propositions we refer to the \texttt{hashB\&B}
search algorithm of Section~\ref{sec:ob3}.

\begin{proposition}
\label{pBestRep}
Let $S\in\intMultiSet$ and let $\tuple{\cost_S,\partial \cost_S,h_S}$
be a base mul equivalent heuristic cost function.  For every base
$B_1$ extracted from the queue and for every base $B_2$ such that
$\prod B_2 = \prod B_1$ then $\cost_S^\alpha B_1 \leq \cost_S^\alpha B_2$.
\end{proposition}

\begin{proof}
Let $\prod B_1 = k$. By complete induction on $k$.

\paragraph{\bf Base:} 

For $k=1$ we have only the empty base and the claim is trivially true.

\paragraph{\bf Step:}  

Assume that the claim is true for every $i<k$ and assume that during
the run of the algorithm we extract a base $B_1$ from the queue with
$\prod B_1=k$. Assume falsely that there exists a base $B_2$ such that
$\prod B_2 = k$ and $\cost_S^\alpha(B_2) < \cost_S^\alpha(B_1)$.This
means that $B_2$ was not evaluated yet (otherwise $B_1=B_2$).
Therefore the father of $B_2$ (in the tree of bases) was never
extracted from the queue.  Let $A_2$ be the closest ancestor of $B_2$
that was extracted. Denote by $C_2$ the child of $A_2$ which is also
the ancestor of $B_2$ (potentially $B_2$ itself). So, $C_2$ was
evaluated.  Observe that $A_2$ and $C_2$ are unique because the search
space (of bases) is a tree and $B_2 \succ C_2$. So,
$\cost_S^\alpha(B_1) \leq \cost_S^\alpha(C_2) \leq
\cost_S^\alpha(B_2)$ and that is a contradiction to the existence of
$B_2$.  \bigskip

\noindent This proves that \textbf{Property 1} derives
\textbf{Property 2}.
\end{proof} 

\begin{proposition}
\label{pBaseExtenstion}
Let $v\in \mathbb{N}$ and $B$,$B'$ bases such that $B'\succ B$.
Then, for $0 \leq j \leq |B|-1$, $v_{(B')}(i) = v_{(B)}(i)$.
\end{proposition}

\noindent\apdxtheoremlike{Theorem~\ref{algIsGood}.}\\
\noindent(1) The $\sumDigits$ cost function satisfies \textbf{Property 1}; and
(2)
the \texttt{hashB\&B} algorithm finds an optimal base for any cost
function which satisfies \textbf{Property~1}.

\begin{proof}~\\
  1) Let $S\in\intMultiSet$, $B_1,B_2 \in \Base$ and $\prod B_1= \prod
  B_2 \in \mathbb{N}$. We prove that if $\cost_S^\alpha(B_1) \leq
  \cost_S^\alpha(B_2)$ then for any base extension $C$, $\cost_S(B_1
  C) \leq \cost_S(B_2 C)$.  The proof is by complete induction on
  $|C|$.

\paragraph{\bf Base:} 

For $|C|=1$, first notice that $h_S(B_1)=h_S(B_2)$. This follows
directly from the definition of admissible heuristics (for the
$\sumDigits$ case). Hence, $\partial \cost_S(B_1) \leq \partial
\cost_S(B_2)$.
From Propositions~\ref{pNib} and~\ref{pBaseExtenstion}, we have that
{\small
\begin{eqnarray*}
  cost_S(B_1 C) &=& \sum_{v \in S} \sum_{i=0}^{|B_1 C|} v_{(B_1 C)}\\
                &=& \sum_{v \in S}\sum_{i=0}^{|B_1|-1}v_{(B_1)} + \sum_{v \in
                     S}\sum_{i=|B_1|}^{|B_1 C|}v_{(B_1 C)}\\
                &=&  \partial \cost_S(B_1) + \sum_{v \in S}  
                      (mod(div(v,\prod B_1),C(0)) + div(v,\prod B_1
                      C))\\
                &=&  \leq \partial \cost_S(B_2)+ \sum_{v \in S} 
                      (mod(div(v,\prod B_2),C(0)) + div(v,\prod B_2
                      C))\\
                &=&  \sum_{v \in S}\sum_{i=0}^{|B_2|-1}v_{(B_2)} + \sum_{v \in
                      S}\sum_{i=|B_2|}^{|B_2 C|}v_{(B_2 C)}\\
                &=&  \sum_{v \in S} \sum_{i=0}^{|B_2 C|} v_{(B_2 C)}\\
                &=&  \cost_S(B_2 C)
\end{eqnarray*}}

\paragraph{\bf Step:}  
For $|C|=k>1$, we define $C=C'\tuple{p}$. So
by
the complete
induction assumption for $|C'|=k-1$ we get that $cost_S(B_1 C') \leq
cost_S(B_2 C')$.  By the fact that $\prod B_1 C' = \prod B_2 C'$ we
can deduce that $msd(S_{(B_1 C')})=msd(S_{(B_2 C')})$. 
By the definition of  admissible heuristics for  $\sumDigits$:
\begin{eqnarray*}
  \partial \cost_S(B_1 C') + \sum msd(S_{(B_1 C')}) &=& 
  cost_S(B_1 C') \leq cost_S(B_2 C') \\
  &=& \partial \cost_S(B_2 C') + \sum msd(S_{(B_2 C')})
\end{eqnarray*}
Therefore, $\partial \cost_S(B_1 C') \leq \partial \cost_S(B_2
C')$.  Combining it with the fact that $h_S(B_1 C')=h_S(B_2 C')$ we
have that $\cost_S^\alpha(B_1 C') \leq \cost_S^\alpha(B_2 C')$. Finally
from the inductive assumption we get that $\cost_S(B_1 C' \tuple{p})
\leq \cost_S(B_2 C' \tuple{p})$.

\medskip
2) Let $\tuple{\cost_S,\partial \cost_S,h_s}$ a base mul equivalent
heuristic cost function and $S\in\intMultiSet$.  We denote by $bestB$
the best base found by the algorithm at each point of the run.  Let
$B$ be the first base extracted from the queue such that
$\cost_S^\alpha(B) \geq \cost_S(bestB)$.  This is the condition that
terminates the main loop of the algorithm, so we need to prove that
$bestB$ is the optimal base for $S$.
 Assume falsely that there exists an optimal base $B'=F.\tuple{b}$ such
 that $\cost_S(B')<\cost_S(bestB)$.  Let $A$ be the nearest ancestor
 of $B'$ such that the its base equivalence class representative $R$
 was extracted from queue ($R\neq B'$ otherwise $\cost_S(bestB) \leq
 \cost_S(B')$).  By Proposition~\ref{pBestRep} we know that
 $\cost_S^\alpha(R)\leq \cost_S^\alpha(A)$ and by \textbf{Property 1} that
 for any base C $\cost_S(RC) \leq \cost_S(AC)$.  In particular for the
 case where $AC=A.\tuple{b}C'=B'$.  By choice of $A$ we get that
 the equivalence class representative of $A.\tuple{b}$ was not
 extracted (and it is the same class of $R.\tuple{b}$).  Therefore,
 $\cost_S(B') \geq \cost_S(R.\tuple{b}C') \geq
 \cost_S^\alpha(R.\tuple{b}) \geq \cost_S^\alpha(B)$, which is a
 contradiction.
\end{proof}

\subsection{On the complexity of the \texttt{hashB\&B} algorithm
  for prime bases}

\begin{theorem}
  Let  $S\in\intMultiSet$ with  $m=max(S)$. Then,
  the complexity of the \texttt{hashB\&B} algorithm for prime bases
  is $O(m^2 * \log(\log(m)))$.
\end{theorem}

\begin{proof}
  We use the prime number theorem which states that the density of the
  primes near $x$ is $O(x/\log(x))$.
  The number of prime bases evaluated in the worst case scenario is: 
  \[O(\sum_{k=1}^{m}{ (m / i) \over \log(m / i)})=
         O(m \sum_{i=1}^{m}{1 \over i \cdot \log(m/i)} )\]
  But 
  \begin{eqnarray*}
    \sum_{i=1}^{m}{1 \over i \cdot  \log(m/i)}&=&
       \sum_{k=1}^{\log(m)}\sum_{i=2^{k-1}}^{2^k}{1 \over i  \cdot \log(m/i)}\\
  &\leq& \sum_{k=1}^{\log(m)}\sum_{i=2^{k-1}}^{2^k}{1 \over 2^{k-1} \cdot( \log(m)- \log(i))}\\
  &=&  \sum_{k=1}^{\log(m)}{2^{k-1} \over 2^{k-1} \cdot( \log(m)- k
    +1)}\\
  &=&  \sum_{k=1}^{\log(m)}{1  \over( \log(m)- k +1)} \\
  &=&\sum_{k=1}^{\log(m)}{1  \over  k}= O(\log(\log(m)))
\end{eqnarray*}
And so the total number of bases evaluated during a run of the
algorithm is bounded by $O(m \cdot \log (\log (m)))$ and the overall
complexity is $O(m^2 \cdot \log(\log(m)))$.
\end{proof}\bigskip

\subsection{Integer division and modulus}
\noindent
Let $a,b\in\mathbb{N}$ with $b>0$. It is standard to define
$div(a,b)$ and $mod(a,b)$ as natural numbers such that $a = div(a,b)
\cdot b + mod(a,b)$ and $0\leq mod(a,b)< b$. In our proofs we note
that $div(a,b)$ is the maximal number such that there exists $r \in
\mathbb{N}$ with $a= div(a,b)\cdot b + r$.

\begin{proposition}
\label{pDiv} 
Let  $a,b,c\in\mathbb{N}$ with $b,c>0$. Then $div(a,b\cdot c) =
div(div(a,b),c)$.
\end{proposition}

 \begin{proof}
   By definition if $k = div(a,b \cdot c)$, $k'=div(a,b)$ and
   $k''=div(k',c)$. Then, $a=k \cdot c \cdot b + r$, $a=k' \cdot b +r'$
   and $k'=k'' \cdot c + r''$. Now, $k \cdot c \leq k'$ because
   otherwise it would be a contradiction to the maximality of $k'$.  If
   $k \cdot c = k'$ then $div(div(a,b),c)=div(k \cdot c, c) =k =
   div(a, b \cdot c)$.  Assume that $k \cdot c < k'$. Then $a=k' \cdot
   b +r' = k'' \cdot c \cdot b + r'' \cdot b +r'$ and from that we
   deduce $k'' \leq k $ (otherwise it would be a contradiction to the
   maximality of $k$).  On the other hand, $0<k' -k \cdot c = k'' \cdot
   c +r'' - k \cdot c$ $\leftrightarrow$ $k \cdot c - r'' < k'' \cdot
   c$. 
   From the definition of modulus
   we get that $r''<c$ and so $(k-1) \cdot c<k \cdot c - r'' < k''
   \cdot c$, and therefore $k-1< k''$.  In total we get that $k-1<k'' \leq
   k$ and because $k$ and $k''$ are natural numbers we get the
   equality.
 \end{proof}\bigskip

\begin{proposition} 
  \label{pDivMod}
  Let $a,b,c\in\mathbb{N}$ such that $b,c>0$. Then, $mod(a,b\cdot
  c) = mod(a,b) + mod(div(a,b),c) \cdot c$.
\end{proposition}

\begin{proof}
  Let $r = mod(a,b \cdot c)$, $ r'=mod(a,b)$, $k'=div(b,c)$, and
  $r''=mod(k',c)$. By definition we can see that $a=k \cdot b \cdot c
  + r$, $a=k' \cdot b +r'$, and $k'=k'' \cdot c + r''$. Therefore
  $mod(a,b \cdot c)=mod(a,b)+mod(div(a,b),c) $ $\leftrightarrow$
  $r=r'+ r'' \cdot b $ $\leftrightarrow$ $a - k \cdot b \cdot c=a-k'
  \cdot b+ (k'-k'' \cdot c) \cdot b $ $\leftrightarrow$ $ - k\cdot
  c=-k' + (k'-k'' \cdot c) $ $\leftrightarrow$ $ k\cdot c=k'' \cdot c
  $ $\leftrightarrow$ $ k=k'' $ $\leftrightarrow$ $div(a,b \cdot c) =
  div(k',c)$ $\leftrightarrow$ $div(a,b \cdot c) = div(div(a,b),c)$
  and this is true by Proposition \ref{pDiv}.
\end{proof}\bigskip

\subsection{The rest of the details}
\noindent
For the sake of readability, we write $weights(B)= \tilde  B$.

\begin{proof}(of Proposition \ref{pNib} by induction on $|B|$). 
  \paragraph{Base (length is 1):} $B=\tuple{b}$ and hence $\tilde B =
  \tuple{1,b}$, $v_{(B)}(0)=mod(v,b), v_{(B)}(1)= div(v,b)$ and by
  definition of div and mod we can see
  that  $v=div(v,b)*b+mod(v,b)$.
  \paragraph{Step:} Assume the assumption is true for every base $B$
  such that $|B|=k-1$.  Let $|B|=k$ such that $B=\tuple{b_0,b_1,\ldots
    ,b_{k-1}}$. Define the base
  $B'=\tuple{b_{0}b_{1},b_2,\ldots,b_{k-1}}$ with $|B'|=k-1$.
  We can see
  that for $0<i < k-1$, $B'(i)=B(i+1)$, and that for $0<i < k$,
  $\tilde B'(i)=\tilde B(i+1)$.  From this we can see that for $0< i
  <k-1$, $v_{(B')}(i)=mod(div(v,\tilde B'(i)),B'(i))=mod(div(v,\tilde
  B(i+1)),B(i+1))=v_{(B)}(i+1)$ and also that
  $v_{(B')}(k-1)=div(v,\tilde B'(k-1))=div(v,\tilde B(k))=v_{(B)}(k)$.
  Back to the main claim by the induction we know that
  $v=\sum_{i=0}^{k-1}v_{(B')}(i) \tilde B'(i)$.
  Therefore $\sum_{i=0}^{k}v_{(B)}(i) \tilde B(i) = v_{(B)}(0) \tilde B(0) +  v_{(B)}(1) \tilde B(1) + \sum_{i=2}^{k}v_{(B)}(i) 	\tilde B(i)= v_{(B)}(0) \tilde B(0) +  v_{(B)}(1) \tilde B(1) + \sum_{i=1}^{k-1}v_{(B')}(i) \tilde B'(i) =  v_{(B)}(0) \tilde B(0) +  v_{(B)}(1) \tilde B(1) + v -  v_{(B')}(0) \tilde B'(0)  $. \\
  By Proposition \ref{pDivMod} we can see that $0= mod(v,b_0) +
  mod(div(v,b_0),b_1)\cdot b_0 - mod(v, b_0 \cdot b_1) = v_{(B)}(0)
  \tilde B(0) + v_{(B)}(1) \tilde B(1) - v_{(B')}(0) \tilde B'(0)$.
  And therefore we get that $\sum_{i=0}^{k}v_{(B)}(i) \tilde B(i) = v
  $.
\end{proof} \bigskip

\begin{proof}(of Proposition \ref{pSubBase1})\\
  First we notice that because $B'\succ B$ then for $0\leq i<|B|$ we
  get that $B(i)=B'(i)$ and therefore $\tilde B(i) = \tilde B'(i)$. By
  Proposition~\ref{pNib} we see that $v_{(B)}(i)=v_{(B')}(i)$.  By
  definition of $S_{(B)}$ and by Proposition~\ref{pNib} and
  the definition of $inputs_S^B$ we prove this proposition by induction on \ $0<
  i \leq|B|$.
	\begin{enumerate}
		\item 
			Base case: $inputs_S^B(1)=\sum_{v \in S} v_{(B)}(0)=\sum_{v \in S} v_{(B')}(0)=inputs_S^{B'}(1)$
		\item 
			Step: Assume that $inputs_S^B(i-1)=inputs_S^{B'}(i-1)$.\\
			$inputs_S^B(i)=\sum_{v \in S}  v_{(B)}(i-1) + div(inputs_S^B(i-1) ,B(i-2)) =$\\
			$\sum_{v \in S}  v_{(B')}(i-1) + div(inputs_S^{B'}(i-1) ,B'(i-2))=inputs_S^{B'}(i)$
	\end{enumerate}
\end{proof} \bigskip

\begin{proof}(of Proposition \ref{pSubBase2})\\
	Let $v$ and $B$ be as defined above. 
	Let $0\leq i \leq |B|$ such that $v\geq \tilde B(i)$. 
	If there exists an $i\leq j \leq |B|$ such that $v=\tilde B(j)$ then 
	we get that \ $div(v,\tilde B(j))=1$ \ and in any case ($j<|B|$ or $j=|B|$) \ \  $v_{(B)}(j)=1$.
	Otherwise let $i \leq j \leq |B|$ be the maximal index such that $v > \tilde B(j)$.
	If $j=|B|$ then $div(v,\tilde B(|B|))>0$. Consider the case when $j<|B|$
	Then $div(v,\tilde B(j))=k$ such that \  $k \cdot \tilde B(j) \leq k \cdot \tilde B(j)+ r=v< \tilde B(j)\cdot B(j)$  
	and so by dividing both sides by $\tilde B(j)$ we get that $div(v,\tilde B(j)) < B(j)$,
	which means that  $v_{(B)}(j)=mod(div(v,\tilde B(j)),B(j))>0$ .
\end{proof} \bigskip

 \begin{proof}(of Proposition \ref{pPrimesOnly})
	 \\The following proof is deeply based on the definition of $v_{(B)}$ and Proposition \ref{pNib}.\\
	Let $v$, $B_1$ and $B_2$ be as defined above.
	\begin{enumerate}
		\item 
			For $0\leq i < p$  by definition  $B_1(i)=B_2(i)$
			which means that $\tilde B_1(i)= \tilde B_2(i)$ and therefore $v_{(B_1)}(i)=v_{(B_2)}(i)$. 
			For $p< i<k-1$  again we get that $B_{2}(i)=B_{1}(i+1)$ 
			and  $\tilde B_{2}(i)=\tilde B_{1}(i+1)$ which again means that $v_{(B_2)}(i)=v_{(B_1)}(i+1)$). 
		\item 
  			We can see that $\tilde B_1(p)=\tilde B_2(p)$.  By Proposition \ref{pNib} we know that \\
			$0=v-v=\sum_{i=0}^{k}v_{B_1}(i)\cdot \tilde B_1(i) - \sum_{i=0}^{k-1}v_{B_2}(i) \cdot \tilde B_2(i)=$\\
			$v_{(B_1)}(p)\cdot \tilde B_1(p)+v_{(B_1)}(p+1)\cdot \tilde B_1(p+1)-v_{(B_2)}(p)\cdot \tilde B_2(p)=$\\
			$v_{(B_1)}(p)\cdot \tilde B_1(p)+v_{(B_1)}(p+1)\cdot \tilde B_1(p)\cdot a-v_{(B_2)}(p)\cdot \tilde B_1(p)$\\
			Therefore \  $v_{(B_2)}(p)=v_{(B_1)}(p)+v_{(B_1)}(p+1)\cdot a$ \\
			 Because $a>1$ we deduce that \ $v_{(B_2)}(p) \geq v_{(B_1)}(p)+v_{(B_1)}(p+1)$ \\
			And in total we get that  \\
			$ \sum_{i=0}^{k-1}v_{(B_2)}(i) - \sum_{i=0}^{k}v_{(B_1)}(i) =  \
			v_{(B_2)}(p) - v_{(B_1)}(p)- v_{(B_1)}(p+1) \geq 0$	
	\end{enumerate}
\end{proof} \bigskip 

\begin{proof}(of Proposition \ref{pBaseExtenstion})\\
Let $v\in \mathbb{N}$ and $B$,$B'$ bases such that $B'\succ B$.
Let $0 \leq i \leq |B|-1$ ($|B|>0$, otherwise there is no such index).
For $i=0$ by Proposition~\ref{pNib} we get that
 $v_{(B')}(0)=mod(v,B'(0))=mod(v,B(0))=v_{(B)}(0)$.
If $i>0$ then by Proposition~\ref{pNib} and definition of $\tilde B$ we can deduce $v_{(B')}(i)=mod(div(v,\tilde B'(i)),B'(i)) =
mod(div(v,\tilde B(i)),B(i))=v_{(B)}(i) $.
\end{proof}

\end{document}